\documentclass[submission,copyright,creativecommons,a4paper]{eptcs}
\providecommand{\event}{QPL 2018} 
\usepackage{breakurl}             
\usepackage{underscore}           

\newcommand{\unbig}[1]{#1}
\usepackage{tikzfig}

\usepackage{amsmath,amsthm,amssymb}
\usepackage{stmaryrd,cmll}

\usepackage{xspace,enumerate,color,epsfig}
\usepackage{graphicx}
\graphicspath{{.}{./figures/}}
\usepackage{marvosym}
\usepackage{bm}

\usepackage{tikzfig}
\usepackage{docmute}
\usepackage{keycommand}

\newkeycommand{\boxmap}[style=small box][1]{\,\tikz{\node[style=\commandkey{style}] (x) {$#1$};\draw(0,-1.25)--(x)--(0,1.25);}\,}
\newkeycommand{\dboxmap}[style=dbox][1]{\,\tikz{\node[style=\commandkey{style}] (x) {$#1$};\draw[boldedge] (0,-1.25)--(x)--(0,1.25);}\,}

\newkeycommand{\wirelabel}[style=white label]{\,\tikz{\node[style=\commandkey{style}]{};}\,\xspace}

\newkeycommand{\pointmap}[style=point][1]{\,\tikz{\node[style=\commandkey{style}] (x) at (0,-0.3) {$#1$};\node [style=none] (3) at (0, 0.9) {};\node [style=none] (3) at (0, -0.9) {};\draw (x)--(0,0.7);}\,}

\newkeycommand{\point}[style=point,estyle=][1]{\,\tikz{\node[style=\commandkey{style}] (x) at (0,-0.05) {$#1$};\draw [\commandkey{estyle}] (x)--(0,1.05);}\,}

\newkeycommand{\copointmap}[style=copoint][1]{\,\tikz{\node[style=\commandkey{style}] (x) at (0,0.3) {$#1$};\draw (x)--(0,-0.7);}\,}

\newkeycommand{\dpoint}[style=point][1]{\,\tikz{\node[style=\commandkey{style},doubled] (x) at (0,-0.3) {$#1$};\draw[boldedge] (x)--(0,0.7);}\,}

\newkeycommand{\kpoint}[style=kpoint,estyle=][1]{\,\tikz{\node[style=\commandkey{style}] (x) at (0,-0.05) {$#1$};\draw [\commandkey{estyle}] (x)--(0,1.05);}\,}

\newkeycommand{\kpointgray}[style=gray kpoint,estyle=][1]{\,\tikz{\node[style=\commandkey{style}] (x) at (0,-0.05) {$#1$};\draw [\commandkey{estyle}] (x)--(0,1.05);}\,}

\newkeycommand{\typedkpoint}[style=kpoint,edgestyle=][2]{%
\,\begin{tikzpicture}
  \begin{pgfonlayer}{nodelayer}
    \node [style=none] (0) at (0, 1) {};
    \node [style=\commandkey{style}] (1) at (0, -0.75) {$#2$};
    \node [style=right label] (2) at (0.25, 0.5) {$#1$};
  \end{pgfonlayer}
  \begin{pgfonlayer}{edgelayer}
    \draw [\commandkey{edgestyle}] (1) to (0.center);
  \end{pgfonlayer}
\end{tikzpicture}\,}

\newkeycommand{\typedkpointdag}[style=kpoint adjoint,edgestyle=][2]{%
\,\begin{tikzpicture}
  \begin{pgfonlayer}{nodelayer}
    \node [style=none] (0) at (0, -1) {};
    \node [style=\commandkey{style}] (1) at (0, 0.75) {$#2$};
    \node [style=right label] (2) at (0.25, -0.5) {$#1$};
  \end{pgfonlayer}
  \begin{pgfonlayer}{edgelayer}
    \draw [\commandkey{edgestyle}] (0.center) to (1);
  \end{pgfonlayer}
\end{tikzpicture}\,}

\newkeycommand{\bistate}[style=kpoint][1]{\,\begin{tikzpicture}
  \begin{pgfonlayer}{nodelayer}
    \node [style=\commandkey{style}, minimum width=1 cm, inner sep=2pt] (0) at (0, -0.25) {$#1$};
    \node [style=none] (1) at (-0.75, 0) {};
    \node [style=none] (2) at (0.75, 0) {};
    \node [style=none] (3) at (-0.75, 0.88) {};
    \node [style=none] (4) at (0.75, 0.88) {};
  \end{pgfonlayer}
  \begin{pgfonlayer}{edgelayer}
    \draw (1.center) to (3.center);
    \draw (2.center) to (4.center);
  \end{pgfonlayer}
\end{tikzpicture}\,}

\newkeycommand{\bistatebig}[style=kpoint][1]{\,\begin{tikzpicture}
  \begin{pgfonlayer}{nodelayer}
    \node [style=\commandkey{style}, minimum width=, minimum width=1.26 cm, inner sep=2pt] (0) at (0, -0.25) {$#1$};
    \node [style=none] (1) at (-1, 0) {};
    \node [style=none] (2) at (1, 0) {};
    \node [style=none] (3) at (-1, 0.88) {};
    \node [style=none] (4) at (1, 0.88) {};
  \end{pgfonlayer}
  \begin{pgfonlayer}{edgelayer}
    \draw (1.center) to (3.center);
    \draw (2.center) to (4.center);
  \end{pgfonlayer}
\end{tikzpicture}\,}

\newkeycommand{\dbistate}[style=dkpoint][1]{\,\begin{tikzpicture}
  \begin{pgfonlayer}{nodelayer}
    \node [style=\commandkey{style}, minimum width=1 cm, inner sep=2pt] (0) at (0, -0.25) {$#1$};
    \node [style=none] (1) at (-0.75, 0) {};
    \node [style=none] (2) at (0.75, 0) {};
    \node [style=none] (3) at (-0.75, 1.0) {};
    \node [style=none] (4) at (0.75, 1.0) {};
  \end{pgfonlayer}
  \begin{pgfonlayer}{edgelayer}
    \draw [boldedge] (1.center) to (3.center);
    \draw [boldedge] (2.center) to (4.center);
  \end{pgfonlayer}
\end{tikzpicture}\,}

\newkeycommand{\bistatebraket}[style1=kpoint,style2=kpointadj][2]{\,\begin{tikzpicture}
  \begin{pgfonlayer}{nodelayer}
    \node [style=\commandkey{style1}, minimum width=1 cm, inner sep=2pt] (0) at (0, -0.75) {$#1$};
    \node [style=none] (1) at (-0.75, -0.5) {};
    \node [style=none] (2) at (0.75, -0.5) {};
    \node [style=none] (3) at (-0.75, 0.5) {};
    \node [style=none] (4) at (0.75, 0.5) {};
    \node [style=\commandkey{style2}, minimum width=1 cm, inner sep=2pt] (5) at (0, 0.75) {$#2$};
  \end{pgfonlayer}
  \begin{pgfonlayer}{edgelayer}
    \draw (1.center) to (3.center);
    \draw (2.center) to (4.center);
  \end{pgfonlayer}
\end{tikzpicture}\,}

\newkeycommand{\dkpoint}[style=dkpoint][1]{\,\tikz{\node[style=\commandkey{style}] (x) at (0,-0.1) {$#1$};\draw[boldedge] (x)--(0,1);}\,}
\newkeycommand{\dkpointadj}[style=dkpointadj][1]{\,\tikz{\node[style=\commandkey{style}] (x) at (0,0.1) {$#1$};\draw[boldedge] (x)--(0,-1);}\,}
\newkeycommand{\dkpointtrans}[style=dkpointtrans][1]{\,\tikz{\node[style=\commandkey{style}] (x) at (0,0.1) {$#1$};\draw[boldedge] (x)--(0,-1);}\,}
\newkeycommand{\dkpointconj}[style=dkpointconj][1]{\,\tikz{\node[style=\commandkey{style}] (x) at (0,-0.1) {$#1$};\draw[boldedge] (x)--(0,1);}\,}

\newcommand{\trace}{\,\begin{tikzpicture}
  \begin{pgfonlayer}{nodelayer}
    \node [style=none] (0) at (0, -0.60) {};
    \node [style=upground] (1) at (0, 0.40) {};
  \end{pgfonlayer}
  \begin{pgfonlayer}{edgelayer}
    \draw (0.center) to (1);
  \end{pgfonlayer}
\end{tikzpicture}\,}

\newcommand\discard\trace

\newcommand{\namedeq}[1]{\,\begin{tikzpicture}
  \begin{pgfonlayer}{nodelayer}
    \node [style=none] (0) at (0, 0.5) {\tiny #1};
    \node [style=none] (1) at (0, 0) {$=$};
  \end{pgfonlayer}
\end{tikzpicture}\,}

\newcommand{\scalareq}{\ensuremath{\propto}\xspace}
\newkeycommand{\pointketbra}[style1=copoint,style2=point][2]{\,%
\begin{tikzpicture}
  \begin{pgfonlayer}{nodelayer}
    \node [style=none] (0) at (0, -1.75) {};
    \node [style=\commandkey{style1}] (1) at (0, -0.75) {$#1$};
    \node [style=\commandkey{style2}] (2) at (0, 0.75) {$#2$};
    \node [style=none] (3) at (0, 1.75) {};
  \end{pgfonlayer}
  \begin{pgfonlayer}{edgelayer}
    \draw (0.center) to (1);
    \draw (2) to (3.center);
  \end{pgfonlayer}
\end{tikzpicture}\,}

\newkeycommand{\twocopointketbra}[style1=copoint,style2=copoint,style3=point][3]{\,%
\begin{tikzpicture}
  \begin{pgfonlayer}{nodelayer}
    \node [style=none] (0) at (-0.75, -1.75) {};
    \node [style=none] (1) at (0.75, -1.75) {};
    \node [style=\commandkey{style1}] (2) at (-0.75, -0.75) {$#1$};
    \node [style=\commandkey{style2}] (3) at (0.75, -0.75) {$#2$};
    \node [style=\commandkey{style3}] (4) at (0, 0.75) {$#3$};
    \node [style=none] (5) at (0, 1.75) {};
  \end{pgfonlayer}
  \begin{pgfonlayer}{edgelayer}
    \draw (0.center) to (2);
    \draw (1.center) to (3);
    \draw (4) to (5.center);
  \end{pgfonlayer}
\end{tikzpicture}\,}

\newkeycommand{\twopointketbra}[style1=copoint,style2=point,style3=point][3]{\,%
\begin{tikzpicture}
  \begin{pgfonlayer}{nodelayer}
    \node [style=none] (0) at (0, -1.75) {};
    \node [style=none] (1) at (-0.75, 1.75) {};
    \node [style=\commandkey{style1}] (2) at (0, -0.75) {$#1$};
    \node [style=\commandkey{style2}] (3) at (-0.75, 0.75) {$#2$};
    \node [style=\commandkey{style3}] (4) at (0.75, 0.75) {$#3$};
    \node [style=none] (5) at (0.75, 1.75) {};
  \end{pgfonlayer}
  \begin{pgfonlayer}{edgelayer}
    \draw (0.center) to (2);
    \draw (1.center) to (3);
    \draw (4) to (5.center);
  \end{pgfonlayer}
\end{tikzpicture}\,}

\newkeycommand{\pointbraket}[style1=point,style2=copoint,estyle=][2]{\,%
\begin{tikzpicture}
  \begin{pgfonlayer}{nodelayer}
    \node [style=\commandkey{style1}] (0) at (0, -0.625) {$#1$};
    \node [style=\commandkey{style2}] (1) at (0, 0.625) {$#2$};
  \end{pgfonlayer}
  \begin{pgfonlayer}{edgelayer}
    \draw [\commandkey{estyle}] (0) to (1);
  \end{pgfonlayer}
\end{tikzpicture}\,}

\newkeycommand{\kpointketbra}[style1=kpointdag,style2=kpoint,estyle=][2]{\,%
\begin{tikzpicture}
  \begin{pgfonlayer}{nodelayer}
    \node [style=none] (0) at (0, -1.9) {};
    \node [style=\commandkey{style1}] (1) at (0, -0.95) {$#1$};
    \node [style=\commandkey{style2}] (2) at (0, 0.95) {$#2$};
    \node [style=none] (3) at (0, 1.9) {};
  \end{pgfonlayer}
  \begin{pgfonlayer}{edgelayer}
    \draw [\commandkey{estyle}] (0.center) to (1);
    \draw [\commandkey{estyle}] (2) to (3.center);
  \end{pgfonlayer}
\end{tikzpicture}\,}

\newkeycommand{\kpointmap}[style1=kpoint,style2=map][2]{\,%
\begin{tikzpicture}
  \begin{pgfonlayer}{nodelayer}
    \node [style=\commandkey{style1}] (0) at (0, -1.1) {$#1$};
    \node [style=\commandkey{style2}] (1) at (0, 0.5) {$#2$};
    \node [style=none] (2) at (0, 1.75) {};
  \end{pgfonlayer}
  \begin{pgfonlayer}{edgelayer}
    \draw (0) to (1);
    \draw (1) to (2.center);
  \end{pgfonlayer}
\end{tikzpicture}%
\,}

\newkeycommand{\dkpointmap}[style1=dkpoint,style2=dmap][2]{\,%
\begin{tikzpicture}
  \begin{pgfonlayer}{nodelayer}
    \node [style=\commandkey{style1}] (0) at (0, -1.2) {$#1$};
    \node [style=\commandkey{style2}] (1) at (0, 0.4) {$#2$};
    \node [style=none] (2) at (0, 1.7) {};
  \end{pgfonlayer}
  \begin{pgfonlayer}{edgelayer}
    \draw[style=boldedge] (0) to (1);
    \draw[style=boldedge] (1) to (2.center);
  \end{pgfonlayer}
\end{tikzpicture}%
\,}

\newkeycommand{\dkpointmapx}[style1=dkpoint,style2=dmap][2]{\,%
\begin{tikzpicture}
  \begin{pgfonlayer}{nodelayer}
    \node [style=\commandkey{style1}] (0) at (0, -1.4) {$#1$};
    \node [style=\commandkey{style2}] (1) at (0, 0.6) {$#2$};
    \node [style=none] (2) at (0, 1.9) {};
  \end{pgfonlayer}
  \begin{pgfonlayer}{edgelayer}
    \draw[style=boldedge] (0) to (1);
    \draw[style=boldedge] (1) to (2.center);
  \end{pgfonlayer}
\end{tikzpicture}%
\,}

\newkeycommand{\kpointsandwichmap}[style1=kpoint,style2=map,style3=kpointdag][3]{\,%
\begin{tikzpicture}
  \begin{pgfonlayer}{nodelayer}
    \node [style=\commandkey{style1}] (0) at (0, -1.5) {$#1$};
    \node [style=\commandkey{style2}] (1) at (0, 0) {$#2$};
    \node [style=\commandkey{style3}] (2) at (0, 1.5) {$#3$};
  \end{pgfonlayer}
  \begin{pgfonlayer}{edgelayer}
    \draw (0) to (1);
    \draw (1) to (2);
  \end{pgfonlayer}
\end{tikzpicture}%
}

\newkeycommand{\sandwichtwo}[style1=point,style2=map,style3=map,style4=copoint][4]{\,%
\begin{tikzpicture}
  \begin{pgfonlayer}{nodelayer}
    \node [style=\commandkey{style2}] (0) at (0, -1) {$#2$};
    \node [style=\commandkey{style3}] (1) at (0, 1) {$#3$};
    \node [style=\commandkey{style1}] (2) at (0, -2.6) {$#1$};
    \node [style=\commandkey{style4}] (3) at (0, 2.6) {$#4$};
  \end{pgfonlayer}
  \begin{pgfonlayer}{edgelayer}
    \draw (2) to (0);
    \draw (0) to (1);
    \draw (1) to (3);
  \end{pgfonlayer}
\end{tikzpicture}\,}

\newkeycommand{\longbraket}[style1=point,style2=copoint][2]{\,%
\begin{tikzpicture}
  \begin{pgfonlayer}{nodelayer}
    \node [style=\commandkey{style1}] (0) at (0, -2.6) {$#1$};
    \node [style=\commandkey{style2}] (1) at (0, 2.6) {$#2$};
  \end{pgfonlayer}
  \begin{pgfonlayer}{edgelayer}
    \draw (0) to (1);
  \end{pgfonlayer}
\end{tikzpicture}\,}

\newkeycommand{\onb}[style=point]{\ensuremath{\{\,\tikz{\node[style=\commandkey{style}] (x) at (0,-0.3) {$j$};\draw (x)--(0,0.7);}\,\}}\xspace}

\newkeycommand{\phase}[style=white phase dot][1]{\,\begin{tikzpicture}
    \begin{pgfonlayer}{nodelayer}
        \node [style=none] (0) at (0, 0.6) {};
        \node [style=\commandkey{style}] (2) at (0, -0) {$#1$}; 
        \node [style=none] (3) at (0, -0.6) {};
    \end{pgfonlayer}
    \begin{pgfonlayer}{edgelayer}
        \draw (2) to (0.center);
        \draw (3.center) to (2);
    \end{pgfonlayer}
\end{tikzpicture}\,}

\newcommand{\hadastate}[1]{\,\tikz{\node[style=hadamard] (x) {$#1$};\draw(x)--(0,0.75);}\,}

\newkeycommand{\dphase}[style=white phase ddot][1]{\,\begin{tikzpicture}
    \begin{pgfonlayer}{nodelayer}
        \node [style=none] (0) at (0, 1.25) {};
        \node [style=\commandkey{style}] (2) at (0, -0) {$#1$};
        \node [style=none] (3) at (0, -1.25) {};
    \end{pgfonlayer}
    \begin{pgfonlayer}{edgelayer}
        \draw [boldedge] (2) to (0.center);
        \draw [boldedge] (3.center) to (2);
    \end{pgfonlayer}
\end{tikzpicture}\,}

\newkeycommand{\dphasepoint}[style=white phase ddot][1]{\,\begin{tikzpicture}[yshift=-3mm]
  \begin{pgfonlayer}{nodelayer}
    \node [style=none] (0) at (0, 1) {};
    \node [style=\commandkey{style}] (1) at (0, 0) {$#1$};
  \end{pgfonlayer}
  \begin{pgfonlayer}{edgelayer}
    \draw [style=boldedge] (0.center) to (1);
  \end{pgfonlayer}
\end{tikzpicture}\,}

\newkeycommand{\dphasepointgray}[style=gray phase ddot][1]{\,\begin{tikzpicture}[yshift=-3mm]
  \begin{pgfonlayer}{nodelayer}
    \node [style=none] (0) at (0, 1) {};
    \node [style=\commandkey{style}] (1) at (0, 0) {$#1$};
  \end{pgfonlayer}
  \begin{pgfonlayer}{edgelayer}
    \draw [style=boldedge] (0.center) to (1);
  \end{pgfonlayer}
\end{tikzpicture}\,}

\newkeycommand{\dphasecopoint}[style=white phase ddot][1]{\,\begin{tikzpicture}[yshift=3mm,yscale=-1]
  \begin{pgfonlayer}{nodelayer}
    \node [style=none] (0) at (0, 1) {};
    \node [style=\commandkey{style}] (1) at (0, 0) {$#1$};
  \end{pgfonlayer}
  \begin{pgfonlayer}{edgelayer}
    \draw [style=boldedge] (0.center) to (1);
  \end{pgfonlayer}
\end{tikzpicture}\,}

\newkeycommand{\dphasepointsm}[style=white phase ddot][1]{\,\begin{tikzpicture}[yshift=-3mm]
  \begin{pgfonlayer}{nodelayer}
    \node [style=none] (0) at (0, 1) {};
    \node [style=\commandkey{style}] (1) at (0, 0) {\footnotesize\!$#1$\!};
  \end{pgfonlayer}
  \begin{pgfonlayer}{edgelayer}
    \draw [style=boldedge] (0.center) to (1);
  \end{pgfonlayer}
\end{tikzpicture}\,}

\newcommand{\greyphase}[1]{\phase[style=grey phase dot]{#1}}

\newkeycommand{\phasepoint}[style=white phase dot][1]{\,\begin{tikzpicture}[yshift=-3mm]
  \begin{pgfonlayer}{nodelayer}
    \node [style=none] (0) at (0, 1) {};
    \node [style=\commandkey{style}] (1) at (0, 0) {$#1$};
  \end{pgfonlayer}
  \begin{pgfonlayer}{edgelayer}
    \draw (0.center) to (1);
  \end{pgfonlayer}
\end{tikzpicture}\,}

\newkeycommand{\phasecopoint}[style=white phase dot][1]{\,\begin{tikzpicture}[yshift=3mm]
  \begin{pgfonlayer}{nodelayer}
    \node [style=none] (0) at (0, -1) {};
    \node [style=\commandkey{style}] (1) at (0, 0) {$#1$};
  \end{pgfonlayer}
  \begin{pgfonlayer}{edgelayer}
    \draw (0.center) to (1);
  \end{pgfonlayer}
\end{tikzpicture}\,}

\newkeycommand{\kpointbraket}[style1=kpoint,style2=kpoint adjoint,estyle=][2]{\,%
\begin{tikzpicture}
  \begin{pgfonlayer}{nodelayer}
    \node [style=\commandkey{style1}] (0) at (0, -0.735) {$#1$};
    \node [style=\commandkey{style2}] (1) at (0, 0.735) {$#2$};
  \end{pgfonlayer}
  \begin{pgfonlayer}{edgelayer}
    \draw [\commandkey{estyle}] (0) to (1);
  \end{pgfonlayer}
\end{tikzpicture}\,}

\newkeycommand{\kpointbraketx}[style1=kpoint,style2=kpoint adjoint,estyle=][2]{\,%
\begin{tikzpicture}
  \begin{pgfonlayer}{nodelayer}
    \node [style=\commandkey{style1}] (0) at (0, -0.885) {$#1$};
    \node [style=\commandkey{style2}] (1) at (0, 0.885) {$#2$};
  \end{pgfonlayer}
  \begin{pgfonlayer}{edgelayer}
    \draw [\commandkey{estyle}] (0) to (1);
  \end{pgfonlayer}
\end{tikzpicture}\,}

\newcommand{\bra}[1]{\ensuremath{\langle #1 |}}
\newcommand{\ket}[1]{\ensuremath{|  #1 \rangle}}

\newcommand{\ketbra}[2]{\ensuremath{\ket{#1}\!\bra{#2}}}

\tikzstyle{cdiag}=[matrix of math nodes, row sep=3em, column sep=3em, text height=1.5ex, text depth=0.25ex,inner sep=0.5em]
\tikzstyle{arrow above}=[transform canvas={yshift=0.5ex}]
\tikzstyle{arrow below}=[transform canvas={yshift=-0.5ex}]

\newcommand{\gendiagram}[1]{
\begin{tikzpicture}
	\begin{pgfonlayer}{nodelayer}
		\node [style=none] (0) at (-0.75, 1) {};
		\node [style=none] (1) at (0.75, 1) {};
		\node [style=none] (2) at (-0.75, -1) {};
		\node [style=none] (3) at (0.75, -1) {};
		\node [style=none] (4) at (0, 0.75) {$\ldots$};
		\node [style=semilarge box] (5) at (0, -0) { #1 };
		\node [style=none] (6) at (0, -0.75) {$\ldots$};
	\end{pgfonlayer}
	\begin{pgfonlayer}{edgelayer}
		\draw (0.center) to (2.center);
		\draw (1.center) to (3.center);
	\end{pgfonlayer}
\end{tikzpicture}
}

\tikzstyle{env}=[copoint,regular polygon rotate=0,minimum width=0.2cm, fill=black]

\tikzstyle{probs}=[shape=semicircle,fill=white,draw=black,shape border rotate=180,minimum width=1.2cm]

\tikzstyle{bbindex}=[font=\color{blue}\footnotesize]

\tikzstyle{nudge}=[yshift=0.6mm]

\tikzstyle{every picture}=[baseline=-0.25em,scale=0.5]
\tikzstyle{dotpic}=[] 
\tikzstyle{diredges}=[every to/.style={diredge}]
\tikzstyle{math matrix}=[matrix of math nodes,left delimiter=(,right delimiter=),inner sep=2pt,column sep=1em,row sep=0.5em,nodes={inner sep=0pt},text height=1.5ex, text depth=0.25ex]

\tikzstyle{inline text}=[text height=1.5ex, text depth=0.25ex,yshift=0.5mm]
\tikzstyle{label}=[font=\footnotesize,text height=1.5ex, text depth=0.25ex]
\tikzstyle{left label}=[label,anchor=east,xshift=2mm]
\tikzstyle{right label}=[label,anchor=west,xshift=-2mm]

\tikzstyle{braceedge}=[decorate,decoration={brace,amplitude=2mm,raise=-1mm}]
\tikzstyle{small braceedge}=[decorate,decoration={brace,amplitude=1mm,raise=-1mm}]

\tikzstyle{doubled}=[line width=1.6pt] 
\tikzstyle{boldedge}=[doubled,shorten <=-0.17mm,shorten >=-0.17mm]
\tikzstyle{boldedgegray}=[doubled,gray,shorten <=-0.17mm,shorten >=-0.17mm]
\tikzstyle{singleedgegray}=[gray]

\tikzstyle{semidoubled}=[line width=1.4pt] 
\tikzstyle{semiboldedgegray}=[semidoubled,gray,shorten <=-0.17mm,shorten >=-0.17mm]

\tikzstyle{boxedge}=[semiboldedgegray]

\tikzstyle{boldedgedashed}=[very thick,dashed,shorten <=-0.17mm,shorten >=-0.17mm]
\tikzstyle{vboldedgedashed}=[doubled,dashed,shorten <=-0.17mm,shorten >=-0.17mm]
\tikzstyle{left hook arrow}=[left hook-latex]
\tikzstyle{right hook arrow}=[right hook-latex]
\tikzstyle{sembracket}=[line width=0.5pt,shorten <=-0.07mm,shorten >=-0.07mm]

\tikzstyle{causal edge}=[->,thick,gray]
\tikzstyle{causal nondir}=[thick,gray]
\tikzstyle{timeline}=[thick,gray, dashed]

\tikzstyle{gray}=[draw=blue!60!white] 

\tikzstyle{cedge}=[<->,thick,gray!70!white]

\tikzstyle{empty diagram}=[draw=gray!40!white,dashed,shape=rectangle,minimum width=1cm,minimum height=1cm]
\tikzstyle{empty diagram small}=[draw=gray!50!white,dashed,shape=rectangle,minimum width=0.6cm,minimum height=0.5cm]

\tikzstyle{dot}=[inner sep=0.3mm,minimum width=2mm,minimum height=2mm,draw,shape=circle,font=\footnotesize]  
\tikzstyle{Wsquare}=[white dot, shape=regular polygon, rounded corners=0.8 mm, minimum size=3.3 mm, regular polygon sides=3, outer sep=-0.2mm]
\tikzstyle{Wsquareadj}=[white dot, shape=regular polygon, rounded corners=0.8 mm, minimum size=3.3 mm, regular polygon sides=3, outer sep=-0.2mm, regular polygon rotate=180]
\tikzstyle{ddot}=[inner sep=0mm, doubled, minimum width=2.5mm,minimum height=2.5mm,draw,shape=circle]

\tikzstyle{black dot}=[dot,fill=black]
\tikzstyle{white dot}=[dot,fill=white,,text depth=-0.2mm]
\tikzstyle{white Wsquare}=[Wsquare,fill=white,,text depth=-0.2mm]
\tikzstyle{white Wsquareadj}=[Wsquareadj,fill=white,,text depth=-0.2mm]
\tikzstyle{green dot}=[white dot] 
\tikzstyle{gray dot}=[dot,fill=gray!40!white,,text depth=-0.2mm]
\tikzstyle{red dot}=[gray dot] 

\tikzstyle{gn}=[circle,fill=green!40!white,draw=black,minimum size=5pt,inner sep=1pt]
\tikzstyle{rn}=[circle,fill=red,font=\color{white},draw=black,minimum size=5pt,inner sep=1pt]
\tikzstyle{Hadamard}=[rectangle,fill=yellow,draw=black,minimum size=8pt,inner sep=0pt]

\tikzstyle{black ddot}=[ddot,fill=black]
\tikzstyle{white ddot}=[ddot,fill=white]
\tikzstyle{gray ddot}=[ddot,fill=gray!40!white]

\tikzstyle{gray edge}=[gray!60!white]

\tikzstyle{small dot}=[inner sep=0.5mm,minimum width=0pt,minimum height=0pt,draw,shape=circle]

\tikzstyle{small black dot}=[small dot,fill=black]
\tikzstyle{small white dot}=[small dot,fill=white]
\tikzstyle{small gray dot}=[small dot,fill=gray!40!white]

\tikzstyle{causal dot}=[inner sep=0.4mm,minimum width=0pt,minimum height=0pt,draw=white,shape=circle,fill=gray!40!white]

\tikzstyle{phase dimensions}=[minimum size=4mm,font=\footnotesize,rectangle,rounded corners=2mm,inner sep=0mm,outer sep=-2mm]
\tikzstyle{dphase dimensions}=[minimum size=5mm,font=\footnotesize,rectangle,rounded corners=2.5mm,inner sep=0.2mm,outer sep=-2mm]

\tikzstyle{white phase dot}=[white dot]

\tikzstyle{white phase ddot}=[ddot,fill=white,dphase dimensions]

\tikzstyle{white rect ddot}=[draw=black,fill=white,doubled,minimum size=5mm,font=\footnotesize,rectangle,rounded corners=2.5mm,inner sep=0.2mm]
\tikzstyle{gray rect ddot}=[draw=black,fill=gray!40!white,doubled,minimum size=6mm,font=\footnotesize,rectangle,rounded corners=3mm]

\tikzstyle{gray phase dot}=[gray dot]

\tikzstyle{gray phase ddot}=[ddot,fill=gray!40!white,dphase dimensions]
\tikzstyle{grey phase dot}=[gray phase dot]
\tikzstyle{grey phase ddot}=[gray phase ddot]

\tikzstyle{small phase dimensions}=[minimum size=4mm,font=\tiny,rectangle,rounded corners=2mm,inner sep=0.2mm,outer sep=-2mm]
\tikzstyle{small dphase dimensions}=[minimum size=4mm,font=\tiny,rectangle,rounded corners=2mm,inner sep=0.2mm,outer sep=-2mm]

\tikzstyle{small gray phase dot}=[dot,fill=gray!40!white,small phase dimensions]
\tikzstyle{small gray phase ddot}=[ddot,fill=gray!40!white,small dphase dimensions]

\tikzstyle{halfscalar}=[star,fill=black,draw=black,minimum size=8pt,inner sep=0pt]

\tikzstyle{small map}=[draw,shape=rectangle,minimum height=4mm,minimum width=4mm,fill=white]

\tikzstyle{cnot}=[fill=white,shape=circle,inner sep=-1.4pt]

\tikzstyle{asym hadamard}=[fill=white,draw,shape=NEbox,inner sep=0.6mm,font=\footnotesize,minimum height=4mm]
\tikzstyle{asym hadamard conj}=[fill=white,draw,shape=NWbox,inner sep=0.6mm,font=\footnotesize,minimum height=4mm]
\tikzstyle{asym hadamard dag}=[fill=white,draw,shape=SEbox,inner sep=0.6mm,font=\footnotesize,minimum height=4mm]

\tikzstyle{hadamard}=[fill=white,draw,inner sep=0.6mm,minimum height=1.5mm,minimum width=1.5mm]
\tikzstyle{small hadamard}=[fill=white,draw,inner sep=0.6mm,minimum height=1.5mm,minimum width=1.5mm]
\tikzstyle{small hadamard rotate}=[small hadamard,rotate=45]
\tikzstyle{dhadamard}=[hadamard,doubled]
\tikzstyle{small dhadamard}=[small hadamard,doubled]
\tikzstyle{small dhadamard rotate}=[small hadamard rotate,doubled]
\tikzstyle{antipode}=[white dot,inner sep=0.3mm,font=\footnotesize]

\tikzstyle{scalar}=[diamond,draw,inner sep=0.5pt,font=\small]
\tikzstyle{dscalar}=[diamond,doubled, draw,inner sep=0.5pt,font=\small]

\tikzstyle{small box}=[rectangle,inline text,fill=white,draw,minimum height=5mm,yshift=-0.5mm,minimum width=5mm,font=\small]
\tikzstyle{small gray box}=[small box,fill=gray!30]
\tikzstyle{medium box}=[rectangle,inline text,fill=white,draw,minimum height=5mm,yshift=-0.5mm,minimum width=8mm,font=\small]
\tikzstyle{square box}=[small box] 
\tikzstyle{medium gray box}=[semilarge box,fill=gray!30]
\tikzstyle{semilarge box}=[rectangle,inline text,fill=white,draw,minimum height=5mm,yshift=-0.5mm,minimum width=12.5mm,font=\small]
\tikzstyle{large box}=[rectangle,inline text,fill=white,draw,minimum height=5mm,yshift=-0.5mm,minimum width=15mm,font=\small]
\tikzstyle{large gray box}=[large box,fill=gray!30]

\tikzstyle{Bayes box}=[rectangle,fill=black,draw, minimum height=3mm, minimum width=3mm]

\tikzstyle{gray square point}=[small box,fill=gray!50]

\tikzstyle{dphase box white}=[dhadamard]
\tikzstyle{dphase box gray}=[dhadamard,fill=gray!50!white]
\tikzstyle{phase box white}=[hadamard]
\tikzstyle{phase box gray}=[hadamard,fill=gray!50!white]

\tikzstyle{point}=[regular polygon,regular polygon sides=3,draw,scale=0.75,inner sep=-0.5pt,minimum width=9mm,fill=white,regular polygon rotate=180]
\tikzstyle{copoint}=[regular polygon,regular polygon sides=3,draw,scale=0.75,inner sep=-0.5pt,minimum width=9mm,fill=white]
\tikzstyle{dpoint}=[point,doubled]
\tikzstyle{dcopoint}=[copoint,doubled]

\tikzstyle{wide copoint}=[fill=white,draw,shape=isosceles triangle,shape border rotate=90,isosceles triangle stretches=true,inner sep=0pt,minimum width=1.5cm,minimum height=6.12mm]
\tikzstyle{wide point}=[fill=white,draw,shape=isosceles triangle,shape border rotate=-90,isosceles triangle stretches=true,inner sep=0pt,minimum width=1.5cm,minimum height=6.12mm,yshift=-0.0mm]
\tikzstyle{wide point plus}=[fill=white,draw,shape=isosceles triangle,shape border rotate=-90,isosceles triangle stretches=true,inner sep=0pt,minimum width=1.74cm,minimum height=7mm,yshift=-0.0mm]

\tikzstyle{wide dpoint}=[fill=white,doubled,draw,shape=isosceles triangle,shape border rotate=-90,isosceles triangle stretches=true,inner sep=0pt,minimum width=1.5cm,minimum height=6.12mm,yshift=-0.0mm]

\tikzstyle{tinypoint}=[regular polygon,regular polygon sides=3,draw,scale=0.55,inner sep=-0.15pt,minimum width=6mm,fill=white,regular polygon rotate=180] 

\tikzstyle{white point}=[point]
\tikzstyle{white dpoint}=[dpoint]
\tikzstyle{green point}=[white point] 
\tikzstyle{white copoint}=[copoint]
\tikzstyle{gray point}=[point,fill=gray!40!white]
\tikzstyle{gray dpoint}=[gray point,doubled]
\tikzstyle{red point}=[gray point] 
\tikzstyle{gray copoint}=[copoint,fill=gray!40!white]
\tikzstyle{gray dcopoint}=[gray copoint,doubled]

\tikzstyle{white point guide}=[regular polygon,regular polygon sides=3,font=\scriptsize,draw,scale=0.65,inner sep=-0.5pt,minimum width=9mm,fill=white,regular polygon rotate=180]

\tikzstyle{black point}=[point,fill=black,font=\color{white}]
\tikzstyle{black copoint}=[copoint,fill=black,font=\color{white}]

\tikzstyle{tiny gray point}=[tinypoint,fill=gray!40!white]

\tikzstyle{diredge}=[->]
\tikzstyle{ddiredge}=[<->]
\tikzstyle{rdiredge}=[<-]
\tikzstyle{thickdiredge}=[->, very thick]
\tikzstyle{pointer edge}=[->,very thick,gray]
\tikzstyle{pointer edge part}=[very thick,gray]
\tikzstyle{dashed edge}=[dashed]
\tikzstyle{thick dashed edge}=[very thick,dashed]
\tikzstyle{thick gray dashed edge}=[thick dashed edge,gray!40]
\tikzstyle{thick map edge}=[very thick,|->]

\makeatletter
\newcommand{\boxshape}[3]{%
\pgfdeclareshape{#1}{
\inheritsavedanchors[from=rectangle] 
\inheritanchorborder[from=rectangle]
\inheritanchor[from=rectangle]{center}
\inheritanchor[from=rectangle]{north}
\inheritanchor[from=rectangle]{south}
\inheritanchor[from=rectangle]{west}
\inheritanchor[from=rectangle]{east}
\backgroundpath{
\southwest \pgf@xa=\pgf@x \pgf@ya=\pgf@y
\northeast \pgf@xb=\pgf@x \pgf@yb=\pgf@y

\@tempdima=#2
\@tempdimb=#3

\pgfpathmoveto{\pgfpoint{\pgf@xa - 5pt + \@tempdima}{\pgf@ya}}
\pgfpathlineto{\pgfpoint{\pgf@xa - 5pt - \@tempdima}{\pgf@yb}}
\pgfpathlineto{\pgfpoint{\pgf@xb + 5pt + \@tempdimb}{\pgf@yb}}
\pgfpathlineto{\pgfpoint{\pgf@xb + 5pt - \@tempdimb}{\pgf@ya}}
\pgfpathlineto{\pgfpoint{\pgf@xa - 5pt + \@tempdima}{\pgf@ya}}
\pgfpathclose
}
}}

\boxshape{NEbox}{0pt}{5pt}
\boxshape{SEbox}{0pt}{-5pt}
\boxshape{NWbox}{5pt}{0pt}
\boxshape{SWbox}{-5pt}{0pt}
\boxshape{EBox}{-3pt}{3pt}
\boxshape{WBox}{3pt}{-3pt}
\makeatother

\tikzstyle{cloud}=[shape=cloud,draw,minimum width=1.5cm,minimum height=1.5cm]

\tikzstyle{map}=[draw,shape=NEbox,inner sep=2pt,minimum height=6mm,fill=white]
\tikzstyle{dashedmap}=[draw,dashed,shape=NEbox,inner sep=2pt,minimum height=6mm,fill=white]
\tikzstyle{mapdag}=[draw,shape=SEbox,inner sep=2pt,minimum height=6mm,fill=white]
\tikzstyle{mapadj}=[draw,shape=SEbox,inner sep=2pt,minimum height=6mm,fill=white]
\tikzstyle{maptrans}=[draw,shape=SWbox,inner sep=2pt,minimum height=6mm,fill=white]
\tikzstyle{mapconj}=[draw,shape=NWbox,inner sep=2pt,minimum height=6mm,fill=white]

\tikzstyle{medium map}=[draw,shape=NEbox,inner sep=2pt,minimum height=6mm,fill=white,minimum width=7mm]
\tikzstyle{medium map dag}=[draw,shape=SEbox,inner sep=2pt,minimum height=6mm,fill=white,minimum width=7mm]
\tikzstyle{medium map adj}=[draw,shape=SEbox,inner sep=2pt,minimum height=6mm,fill=white,minimum width=7mm]
\tikzstyle{medium map trans}=[draw,shape=SWbox,inner sep=2pt,minimum height=6mm,fill=white,minimum width=7mm]
\tikzstyle{medium map conj}=[draw,shape=NWbox,inner sep=2pt,minimum height=6mm,fill=white,minimum width=7mm]
\tikzstyle{semilarge map}=[draw,shape=NEbox,inner sep=2pt,minimum height=6mm,fill=white,minimum width=9.5mm]
\tikzstyle{semilarge map trans}=[draw,shape=SWbox,inner sep=2pt,minimum height=6mm,fill=white,minimum width=9.5mm]
\tikzstyle{semilarge map adj}=[draw,shape=SEbox,inner sep=2pt,minimum height=6mm,fill=white,minimum width=9.5mm]
\tikzstyle{semilarge map dag}=[draw,shape=SEbox,inner sep=2pt,minimum height=6mm,fill=white,minimum width=9.5mm]
\tikzstyle{semilarge map conj}=[draw,shape=NWbox,inner sep=2pt,minimum height=6mm,fill=white,minimum width=9.5mm]
\tikzstyle{large map}=[draw,shape=NEbox,inner sep=2pt,minimum height=6mm,fill=white,minimum width=12mm]
\tikzstyle{large map conj}=[draw,shape=NWbox,inner sep=2pt,minimum height=6mm,fill=white,minimum width=12mm]
\tikzstyle{very large map}=[draw,shape=NEbox,inner sep=2pt,minimum height=6mm,fill=white,minimum width=17mm]

\tikzstyle{medium dmap}=[draw,doubled,shape=NEbox,inner sep=2pt,minimum height=6mm,fill=white,minimum width=7mm]
\tikzstyle{medium dmap dag}=[draw,doubled,shape=SEbox,inner sep=2pt,minimum height=6mm,fill=white,minimum width=7mm]
\tikzstyle{medium dmap adj}=[draw,doubled,shape=SEbox,inner sep=2pt,minimum height=6mm,fill=white,minimum width=7mm]
\tikzstyle{medium dmap trans}=[draw,doubled,shape=SWbox,inner sep=2pt,minimum height=6mm,fill=white,minimum width=7mm]
\tikzstyle{medium dmap conj}=[draw,doubled,shape=NWbox,inner sep=2pt,minimum height=6mm,fill=white,minimum width=7mm]
\tikzstyle{semilarge dmap}=[draw,doubled,shape=NEbox,inner sep=2pt,minimum height=6mm,fill=white,minimum width=9.5mm]
\tikzstyle{semilarge dmap trans}=[draw,doubled,shape=SWbox,inner sep=2pt,minimum height=6mm,fill=white,minimum width=9.5mm]
\tikzstyle{semilarge dmap adj}=[draw,doubled,shape=SEbox,inner sep=2pt,minimum height=6mm,fill=white,minimum width=9.5mm]
\tikzstyle{semilarge dmap dag}=[draw,doubled,shape=SEbox,inner sep=2pt,minimum height=6mm,fill=white,minimum width=9.5mm]
\tikzstyle{semilarge dmap conj}=[draw,doubled,shape=NWbox,inner sep=2pt,minimum height=6mm,fill=white,minimum width=9.5mm]
\tikzstyle{large dmap}=[draw,doubled,shape=NEbox,inner sep=2pt,minimum height=6mm,fill=white,minimum width=12mm]
\tikzstyle{large dmap conj}=[draw,doubled,shape=NWbox,inner sep=2pt,minimum height=6mm,fill=white,minimum width=12mm]
\tikzstyle{large dmap trans}=[draw,doubled,shape=SWbox,inner sep=2pt,minimum height=6mm,fill=white,minimum width=12mm]
\tikzstyle{large dmap adj}=[draw,doubled,shape=SEbox,inner sep=2pt,minimum height=6mm,fill=white,minimum width=12mm]
\tikzstyle{large dmap dag}=[draw,doubled,shape=SEbox,inner sep=2pt,minimum height=6mm,fill=white,minimum width=12mm]
\tikzstyle{very large dmap}=[draw,doubled,shape=NEbox,inner sep=2pt,minimum height=6mm,fill=white,minimum width=19.5mm]

\tikzstyle{muxbox}=[draw,shape=rectangle,minimum height=3mm,minimum width=3mm,fill=white]
\tikzstyle{dmuxbox}=[muxbox,doubled]

\tikzstyle{box}=[draw,shape=rectangle,inner sep=2pt,minimum height=6mm,minimum width=6mm,fill=white]
\tikzstyle{dbox}=[draw,doubled,shape=rectangle,inner sep=2pt,minimum height=6mm,minimum width=6mm,fill=white]
\tikzstyle{dmap}=[draw,doubled,shape=NEbox,inner sep=2pt,minimum height=6mm,fill=white]
\tikzstyle{dmapdag}=[draw,doubled,shape=SEbox,inner sep=2pt,minimum height=6mm,fill=white]
\tikzstyle{dmapadj}=[draw,doubled,shape=SEbox,inner sep=2pt,minimum height=6mm,fill=white]
\tikzstyle{dmaptrans}=[draw,doubled,shape=SWbox,inner sep=2pt,minimum height=6mm,fill=white]
\tikzstyle{dmapconj}=[draw,doubled,shape=NWbox,inner sep=2pt,minimum height=6mm,fill=white]

\tikzstyle{ddmap}=[draw,doubled,dashed,shape=NEbox,inner sep=2pt,minimum height=6mm,fill=white]
\tikzstyle{ddmapdag}=[draw,doubled,dashed,shape=SEbox,inner sep=2pt,minimum height=6mm,fill=white]
\tikzstyle{ddmapadj}=[draw,doubled,dashed,shape=SEbox,inner sep=2pt,minimum height=6mm,fill=white]
\tikzstyle{ddmaptrans}=[draw,doubled,dashed,shape=SWbox,inner sep=2pt,minimum height=6mm,fill=white]
\tikzstyle{ddmapconj}=[draw,doubled,dashed,shape=NWbox,inner sep=2pt,minimum height=6mm,fill=white]

\boxshape{sNEbox}{0pt}{3pt}
\boxshape{sSEbox}{0pt}{-3pt}
\boxshape{sNWbox}{3pt}{0pt}
\boxshape{sSWbox}{-3pt}{0pt}
\tikzstyle{smap}=[draw,shape=sNEbox,fill=white]
\tikzstyle{smapdag}=[draw,shape=sSEbox,fill=white]
\tikzstyle{smapadj}=[draw,shape=sSEbox,fill=white]
\tikzstyle{smaptrans}=[draw,shape=sSWbox,fill=white]
\tikzstyle{smapconj}=[draw,shape=sNWbox,fill=white]

\tikzstyle{dsmap}=[draw,dashed,shape=sNEbox,fill=white]
\tikzstyle{dsmapdag}=[draw,dashed,shape=sSEbox,fill=white]
\tikzstyle{dsmaptrans}=[draw,dashed,shape=sSWbox,fill=white]
\tikzstyle{dsmapconj}=[draw,dashed,shape=sNWbox,fill=white]

\boxshape{mNEbox}{0pt}{10pt}
\boxshape{mSEbox}{0pt}{-10pt}
\boxshape{mNWbox}{10pt}{0pt}
\boxshape{mSWbox}{-10pt}{0pt}
\tikzstyle{mmap}=[draw,shape=mNEbox]
\tikzstyle{mmapdag}=[draw,shape=mSEbox]
\tikzstyle{mmaptrans}=[draw,shape=mSWbox]
\tikzstyle{mmapconj}=[draw,shape=mNWbox]

\tikzstyle{mmapgray}=[draw,fill=gray!40!white,shape=mNEbox]
\tikzstyle{smapgray}=[draw,fill=gray!40!white,shape=sNEbox]

\makeatletter

\pgfdeclareshape{cornerpoint}{
\inheritsavedanchors[from=rectangle] 
\inheritanchorborder[from=rectangle]
\inheritanchor[from=rectangle]{center}
\inheritanchor[from=rectangle]{north}
\inheritanchor[from=rectangle]{south}
\inheritanchor[from=rectangle]{west}
\inheritanchor[from=rectangle]{east}
\backgroundpath{
\southwest \pgf@xa=\pgf@x \pgf@ya=\pgf@y
\northeast \pgf@xb=\pgf@x \pgf@yb=\pgf@y

\pgfmathsetmacro{\pgf@shorten@left}{\pgfkeysvalueof{/tikz/shorten left}}
\pgfmathsetmacro{\pgf@shorten@right}{\pgfkeysvalueof{/tikz/shorten right}}

\pgfpathmoveto{\pgfpoint{0.5 * (\pgf@xa + \pgf@xb)}{\pgf@ya - 5pt}}
\pgfpathlineto{\pgfpoint{\pgf@xa - 8pt + \pgf@shorten@left}{\pgf@yb - 1.5 * \pgf@shorten@left}}
\pgfpathlineto{\pgfpoint{\pgf@xa - 8pt + \pgf@shorten@left}{\pgf@yb}}
\pgfpathlineto{\pgfpoint{\pgf@xb + 8pt - \pgf@shorten@right}{\pgf@yb}}
\pgfpathlineto{\pgfpoint{\pgf@xb + 8pt - \pgf@shorten@right}{\pgf@yb - 1.5 * \pgf@shorten@right}}
\pgfpathclose
}
}

\pgfdeclareshape{cornercopoint}{
\inheritsavedanchors[from=rectangle] 
\inheritanchorborder[from=rectangle]
\inheritanchor[from=rectangle]{center}
\inheritanchor[from=rectangle]{north}
\inheritanchor[from=rectangle]{south}
\inheritanchor[from=rectangle]{west}
\inheritanchor[from=rectangle]{east}
\backgroundpath{
\southwest \pgf@xa=\pgf@x \pgf@ya=\pgf@y
\northeast \pgf@xb=\pgf@x \pgf@yb=\pgf@y

\pgfmathsetmacro{\pgf@shorten@left}{\pgfkeysvalueof{/tikz/shorten left}}
\pgfmathsetmacro{\pgf@shorten@right}{\pgfkeysvalueof{/tikz/shorten right}}

\pgfpathmoveto{\pgfpoint{0.5 * (\pgf@xa + \pgf@xb)}{\pgf@yb + 5pt}}
\pgfpathlineto{\pgfpoint{\pgf@xa - 8pt + \pgf@shorten@left}{\pgf@ya + 1.5 * \pgf@shorten@left}}
\pgfpathlineto{\pgfpoint{\pgf@xa - 8pt + \pgf@shorten@left}{\pgf@ya}}
\pgfpathlineto{\pgfpoint{\pgf@xb + 8pt - \pgf@shorten@right}{\pgf@ya}}
\pgfpathlineto{\pgfpoint{\pgf@xb + 8pt - \pgf@shorten@right}{\pgf@ya + 1.5 * \pgf@shorten@right}}
\pgfpathclose
}
}

\makeatother

\pgfkeyssetvalue{/tikz/shorten left}{0pt}
\pgfkeyssetvalue{/tikz/shorten right}{0pt}

\tikzstyle{kpoint common}=[draw,fill=white,inner sep=1pt,minimum height=4mm]
\tikzstyle{kpoint sc}=[shape=cornerpoint,kpoint common]
\tikzstyle{kpoint adjoint sc}=[shape=cornercopoint,kpoint common]
\tikzstyle{kpoint}=[shape=cornerpoint,shorten left=5pt,kpoint common]
\tikzstyle{kpoint adjoint}=[shape=cornercopoint,shorten left=5pt,kpoint common]
\tikzstyle{kpoint conjugate}=[shape=cornerpoint,shorten right=5pt,kpoint common]
\tikzstyle{kpoint transpose}=[shape=cornercopoint,shorten right=5pt,kpoint common]
\tikzstyle{kpoint symm}=[shape=cornerpoint,shorten left=5pt,shorten right=5pt,kpoint common]

\tikzstyle{black kpoint}=[shape=cornerpoint,shorten left=5pt,kpoint common,fill=black,font=\color{white}]
\tikzstyle{black kpoint adjoint}=[shape=cornercopoint,shorten left=5pt,kpoint common,fill=black,font=\color{white}]
\tikzstyle{black kpointadj}=[shape=cornercopoint,shorten left=5pt,kpoint common,fill=black,font=\color{white}]

\tikzstyle{black dkpoint}=[shape=cornerpoint,shorten left=5pt,kpoint common,fill=black, doubled,font=\color{white}]
\tikzstyle{black dkpoint adjoint}=[shape=cornercopoint,shorten left=5pt,kpoint common,fill=black, doubled,font=\color{white}]
\tikzstyle{black dkpointadj}=[shape=cornercopoint,shorten left=5pt,kpoint common,fill=black, doubled,font=\color{white}] 

\tikzstyle{kpointdag}=[kpoint adjoint]
\tikzstyle{kpointadj}=[kpoint adjoint]
\tikzstyle{kpointconj}=[kpoint conjugate]
\tikzstyle{kpointtrans}=[kpoint transpose]

\tikzstyle{big kpoint}=[kpoint, minimum width=1.2 cm, minimum height=8mm, inner sep=4pt, text depth=3mm]

\tikzstyle{wide kpoint}=[kpoint, minimum width=1 cm, inner sep=2pt]
\tikzstyle{wide kpointdag}=[kpointdag, minimum width=1 cm, inner sep=2pt]
\tikzstyle{wide kpointconj}=[kpointconj, minimum width=1 cm, inner sep=2pt]
\tikzstyle{wide kpointtrans}=[kpointtrans, minimum width=1 cm, inner sep=2pt]

\tikzstyle{gray kpoint}=[kpoint,fill=gray!50!white]
\tikzstyle{gray kpointdag}=[kpointdag,fill=gray!50!white]
\tikzstyle{gray kpointadj}=[kpointadj,fill=gray!50!white]
\tikzstyle{gray kpointconj}=[kpointconj,fill=gray!50!white]
\tikzstyle{gray kpointtrans}=[kpointtrans,fill=gray!50!white]

\tikzstyle{gray dkpoint}=[kpoint,fill=gray!50!white,doubled]
\tikzstyle{gray dkpointdag}=[kpointdag,fill=gray!50!white,doubled]
\tikzstyle{gray dkpointadj}=[kpointadj,fill=gray!50!white,doubled]
\tikzstyle{gray dkpointconj}=[kpointconj,fill=gray!50!white,doubled]
\tikzstyle{gray dkpointtrans}=[kpointtrans,fill=gray!50!white,doubled]

\tikzstyle{white label}=[draw,fill=white,rectangle,inner sep=0.7 mm]
\tikzstyle{gray label}=[draw,fill=gray!50!white,rectangle,inner sep=0.7 mm]
\tikzstyle{black label}=[draw,fill=black,rectangle,inner sep=0.7 mm]

\tikzstyle{dkpoint}=[kpoint,doubled]
\tikzstyle{wide dkpoint}=[wide kpoint,doubled]
\tikzstyle{dkpointdag}=[kpoint adjoint,doubled]
\tikzstyle{wide dkpointdag}=[wide kpointdag,doubled]
\tikzstyle{dkcopoint}=[kpoint adjoint,doubled]
\tikzstyle{dkpointadj}=[kpoint adjoint,doubled]
\tikzstyle{dkpointconj}=[kpoint conjugate,doubled]
\tikzstyle{dkpointtrans}=[kpoint transpose,doubled]

\tikzstyle{kscalar}=[kpoint common, shape=EBox, inner xsep=-1pt, inner ysep=3pt,font=\small]
\tikzstyle{kscalarconj}=[kpoint common, shape=WBox, inner xsep=-1pt, inner ysep=3pt,font=\small]

\tikzstyle{spekpoint}=[kpoint sc,minimum height=5mm,inner sep=3pt]
\tikzstyle{spekcopoint}=[kpoint adjoint sc,minimum height=5mm,inner sep=3pt]

\tikzstyle{dspekpoint}=[spekpoint,doubled]
\tikzstyle{dspekcopoint}=[spekcopoint,doubled]

 \tikzstyle{upground}=[circuit ee IEC,thick,ground,rotate=90,scale=1.5,inner sep=-2mm]
 \tikzstyle{downground}=[circuit ee IEC,thick,ground,rotate=-90,scale=1.5,inner sep=-2mm]

\tikzstyle{maxmix}=[regular polygon,regular polygon sides=3,draw=black,xscale=0.4,yscale=0.3,inner sep=-0.5pt,minimum width=10mm,fill=gray,regular polygon rotate=180]

 \tikzstyle{bigground}=[regular polygon,regular polygon sides=3,draw=gray,scale=0.50,inner sep=-0.5pt,minimum width=10mm,fill=gray]

\tikzstyle{arrs}=[-latex,font=\small,auto]
\tikzstyle{arrow plain}=[arrs]
\tikzstyle{arrow dashed}=[dashed,arrs]
\tikzstyle{arrow bold}=[very thick,arrs]
\tikzstyle{arrow hide}=[draw=white!0,-]
\tikzstyle{arrow reverse}=[latex-]
\tikzstyle{cdnode}=[]

\tikzstyle{Z}=[white dot]
\tikzstyle{wire}=[none]
\tikzstyle{X}=[gray dot]
\tikzstyle{bbox}=[]
\tikzstyle{simple}=[]
\tikzstyle{quanto}=[]

\newcommand{\dotcounit}[1]{%
\,\begin{tikzpicture}[dotpic,yshift=-1mm]
\node [#1] (a) at (0,0.35) {}; 
\draw (0,-0.3)--(a);
\end{tikzpicture}\,\xspace}
\newcommand{\dotunit}[1]{%
\,\begin{tikzpicture}[dotpic,yshift=1.5mm]
\node [#1] (a) at (0,-0.35) {}; 
\draw (a)--(0,0.3);
\end{tikzpicture}\,\xspace}

\newcommand{\dotmult}[1]{%
\,\begin{tikzpicture}[dotpic]
	\node [#1] (a) {};
	\draw (a) -- (90:0.55);
	\draw (a) (-45:0.6) -- (a);
	\draw (a) (-135:0.6) -- (a);
\end{tikzpicture}\,\xspace}

\newcommand{\dotonly}[1]{%
\,\begin{tikzpicture}[dotpic]
\node [#1] (a) at (0,0) {};
\end{tikzpicture}\,}

\newcommand{\whitedot}{\dotonly{white dot}\xspace}
\newcommand{\whiteunit}{\dotunit{white dot}}
\newcommand{\whitecounit}{\dotcounit{white dot}}
\newcommand{\whitemult}{\dotmult{white dot}}

\newcommand{\graymult}{\dotmult{gray dot}}

\newcommand{\hadaunit}{\dotunit{small hadamard}}

\let\olddagger\dagger
\renewcommand{\dagger}{\ensuremath{\olddagger}\xspace}

\theoremstyle{definition}
\newtheorem{theorem}{Theorem}[section]
\newtheorem{corollary}[theorem]{Corollary}
\newtheorem{lemma}[theorem]{Lemma}
\newtheorem{proposition}[theorem]{Proposition}

\newtheorem{example*}[theorem]{Example*}
\newtheorem{examples*}[theorem]{Examples*}
\newtheorem{remark}[theorem]{Remark}
\newtheorem{remark*}[theorem]{Remark*}

\newtheorem*{theorem*}{Theorem}
\newtheorem*{corollary*}{Corollary}
\newtheorem*{lemma*}{Lemma}
\newtheorem*{proposition*}{Proposition}

\usepackage[color]{changebar}

\usepackage{color}
\def\bR{\begin{color}{red}} 
\def\bB{\begin{color}{blue}}
\def\bM{\begin{color}{magenta}}
\def\bC{\begin{color}{cyan}}
\def\bW{\begin{color}{white}}
\def\bBl{\begin{color}{black}} 
\def\bG{\begin{color}{green}}
\def\bY{\begin{color}{yellow}}
\def\e{\end{color}\xspace}
\newcommand{\bit}{\begin{itemize}}
\newcommand{\eit}{\end{itemize}\par\noindent}
\newcommand{\ben}{\begin{enumerate}}
\newcommand{\een}{\end{enumerate}\par\noindent}
\newcommand{\beq}{\begin{equation}}
\newcommand{\eeq}{\end{equation}\par\noindent}
\newcommand{\beqa}{\begin{eqnarray*}}
\newcommand{\eeqa}{\end{eqnarray*}\par\noindent}
\newcommand{\beqn}{\begin{eqnarray}}
\newcommand{\eeqn}{\end{eqnarray}\par\noindent}

\newcommand{\intf}[1]{\left\llbracket #1 \right\rrbracket} 

\title{ZH: A Complete Graphical Calculus for Quantum Computations Involving Classical Non-linearity}
\author{Miriam Backens
\institute{Department of Computer Science\\
University of Oxford}
\email{miriam.backens@cs.ox.ac.uk}
\and
Aleks Kissinger
\institute{Institute for Computing and Information Sciences\\
Radboud University}
\email{aleks@cs.ru.nl}
}
\def\titlerunning{The ZH Calculus}
\def\authorrunning{M. Backens \& A. Kissinger}
\begin{document}

\maketitle

\begin{abstract}
We present a new graphical calculus that is sound and complete for a universal family of quantum circuits, which can be seen as the natural string-diagrammatic extension of the approximately (real-valued) universal family of Hadamard+CCZ circuits. The diagrammatic language is generated by two kinds of nodes: the so-called `spider' associated with the computational basis, as well as a new arity-$N$ generalisation of the Hadamard gate, which satisfies a variation of the spider fusion law. Unlike previous graphical calculi, this admits compact encodings of non-linear classical functions. For example, the AND gate can be depicted as a diagram of just 2 generators, compared to $\sim25$ in the ZX-calculus. Consequently, $N$-controlled gates, hypergraph states, Hadamard+Toffoli circuits, and diagonal circuits at arbitrary levels of the Clifford hierarchy also enjoy encodings with low constant overhead. This suggests that this calculus will be significantly more convenient for reasoning about the interplay between classical non-linear behaviour (e.g. in an oracle) and purely quantum operations. After presenting the calculus, we will prove it is sound and complete for universal quantum computation by demonstrating the reduction of any diagram to an easily describable normal form.
\end{abstract}

\section{Introduction}

Graphical calculi enable reasoning about quantum computation in an intuitive yet rigorous way.
Calculi based on string diagrams are more flexible than circuit-style languages, allowing the description of states and measurement projections as well as unitary operations in one unified framework.
Their rigour is ensured by the category-theoretical underpinnings \cite{SelingerCPM}.
The best-known of these graphical languages is the ZX-calculus, which was first introduced 10 years ago \cite{CD1,CD2}.
It is built around the interactions of two complementary bases, the computational basis and the Hadamard basis, which are graphically represented by so-called \emph{spiders}.
A related formalism is the ZW-calculus \cite{hadzihasanovic2017thesis}, which is built around the interactions of generators related to the two different types of three-qubit entangled states: GHZ states and $W$ states.

Here, we introduce a new graphical language called the \textit{ZH-calculus}, which roughly follows this analogy with multipartite entanglement:
\begin{center}
	\textit{ZX-calculus} : \textit{ZH-calculus} ::
	\textit{graph states} : \textit{hypergraph states}
\end{center}
Graph states are the basic resource for the one-way model of measurement-based quantum computation~\cite{MBQC2}, and have been studied extensively using the ZX-calculus~\cite{CD2,DP1,DP2,RossMBQC}.
Hypergraph states were introduced in~\cite{rossi2013hypergraph} as a generalisation of graph states, and have recently gathered some interest due, for example, to the role they play in quantum search algorithms~\cite{HyperGrover}, exponentially-growing Bell violations~\cite{gachechiladze2016extreme}, and universal measurement-based quantum computation~\cite{HyperSPTO}.

Like the ZX- and ZW-calculi, the ZH-calculus includes a family of ``Z-spiders'' associated with the computational basis. However, its point of departure is the inclusion of ``H-boxes'', which are $n$-ary generalisations of the Hadamard gate satisfying a variation of the spider fusion law, much like the one satisfied by $W$-spiders in the ZW-calculus.\footnote{Despite satisfying a similar variation of the spider fusion rule, this generalisation of the Hadamard node is different from that employed in the original version of the Y-calculus \cite[Definition 2 of Version 1]{jeandel2018y-calculus}.} Whereas Hadamard gates are used to introduce edges between 2 vertices in a graph state, H-boxes can introduce hyperedges between $n$ vertices in a hypergraph state. Seen from another perspective, H-boxes are closely related to both $n$-ary AND gates in classical logic and to the Hadamard-CCZ family of quantum circuits.
As a result, Boolean circuits can be encoded in the ZH-calculus with low constant overhead. In particular, the linear maps corresponding to classical AND and NOT gates can be depicted as follows in terms of the ZH calculus:
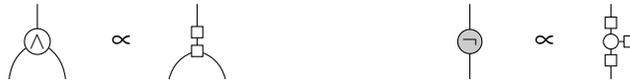
\ctikzfig{logic}
While the unitary NOT gate has a simple expression in the ZX-calculus, a typical encoding of an AND gate requires $25$ basic generators and non-Clifford phases (cf.~\cite{CKbook}, \S12.1.3.1).
Similarly, multiply-controlled phase gates also have very succinct representations, indicating that the ZH-calculus may be useful for analysing Hadamard-CCZ circuits (a.k.a. Hadamard-Toffoli circuits~\cite{ShiToffoli,aharonov2003hadamardtoffoli}, cf. forthcoming~\cite{Niel2018} for connection to ZH), as well as diagonal operators at any level of the Clifford hierarchy~\cite{DiagHierarchy}.

Our main theorem is the ZH-calculus is complete with respect to its standard interpretation as matrices. That is, if two ZH-diagrams describe the same matrices, then they are provably equal using the rules of the ZH-calculus. Perhaps one of the most appealing features of the calculus is the simplicity of this completeness proof. The core of the proof (section~\ref{s:completeness}) fits on 4 pages, where only especially intricate lemmas---which appear in Appendix~\ref{sec:disconnect}---were done within the proof assistant Quantomatic~\cite{quanto-cade}. This is due to two main factors. The first is the extensive use of \textit{!-box notation}~\cite{kissinger2014pattern}, which gives an elegant way to reason about diagrams which have arbitrarily-large fan-out-type behaviour. The second is a unique normal form for the ZH-calculus, which expresses any matrix as a Schur product -- i.e.\ entrywise product -- of elementary matrices with the property that all but one entry of each matrix is 1.
This multiplicative construction contrasts with the additive construction of the normal form in the ZW-calculus \cite{hadzihasanovic2017thesis}, which arises as a sum of elementary matrices with the property that all but one entry of each matrix is 0.
For example the normal form of the diagram corresponding to the matrix $\left(\begin{smallmatrix}a&b\\c&d\end{smallmatrix}\right)$ is effectively constructed as follows, where the left-hand side represents the approach in the ZH-calculus with $*$ denoting entrywise multiplication, and the right-hand side represents the approach in the ZW-calculus:
\[
 \begin{pmatrix}a&1\\1&1\end{pmatrix} * \begin{pmatrix}1&b\\1&1\end{pmatrix} * \begin{pmatrix}1&1\\c&1\end{pmatrix} * \begin{pmatrix}1&1\\1&d\end{pmatrix}
 = \begin{pmatrix}a&b\\c&d\end{pmatrix}
 = \begin{pmatrix}a&0\\0&0\end{pmatrix} + \begin{pmatrix}0&b\\0&0\end{pmatrix} + \begin{pmatrix}0&0\\c&0\end{pmatrix} + \begin{pmatrix}0&0\\0&d\end{pmatrix}.
\]
Unlike the completeness proofs for universal versions of the ZX-calculus \cite{LoriaCompleteness,OxfordCompleteness}, which make use of the ZW-completeness proof via suitable translations between the two respective languages, our proofs of soundness, completeness, and universality are self-contained and don't rely on encoding into another calculus.

The paper is structured as follows. The generators and relations of the ZH-calculus are introduced in Section~\ref{s:ZH-dfn} and the calculus is proved to be universal and sound. The completeness proof is given in Section~\ref{s:completeness}. In section~\ref{s:applications} we survey two potential applications and comment on future work. Omitted proofs and a link to the Quantomatic project used to prove Lemmas~\ref{lem:disconnect-4} and \ref{lem:big-disconnect} are given in the appendix.

\section{Definition of the ZH-calculus}
\label{s:ZH-dfn}

The ZH-calculus is a graphical language expressing operations as \emph{string diagrams}.
These are diagrams consisting of dots or boxes, connected by wires.
Wires are also allowed to have one or two ``dangling'' ends, which are not connected to a dot or box: these represent inputs of the diagram if oriented towards the bottom, outputs of the diagram if oriented to the top.

\subsection{The generators and their interpretation}
\label{s:ZX-translation}

The diagrams of the ZH-calculus are generated by \emph{Z-spiders}, which are represented as white dots, and \emph{H-boxes}, which are represented as white boxes labelled with a complex number $a$.
These generators are interpreted as follows, where $\intf{\cdot}$ denotes the map from diagrams to matrices.
\[
 \intf{\input{Z-spider.tikz}} := \ket{0}^{\otimes n}\bra{0}^{\otimes m} + \ket{1}^{\otimes n}\bra{1}^{\otimes m} \qquad\qquad
 \intf{\input{H-spider.tikz}} := \sum a^{i_1\ldots i_m j_1\ldots j_n} \ket{j_1\ldots j_n}\bra{i_1\ldots i_m}
\]
The sum in the second equation is over all $i_1,\ldots, i_m, j_1,\ldots, j_n\in\{0,1\}$, i.e.\ an H-box represents a matrix all but one of whose entries are equal to 1.
A label of $-1$ is usually left out and the box is then drawn smaller, e.g.\ $\hadaunit:=\hadastate{-1}$.
Straight and curved wires have the following interpretations:
\[
 \intf{\;|\;} := \ketbra{0}{0}+\ketbra{1}{1} \qquad\qquad\qquad
 \intf{\begin{tikzpicture}
	\begin{pgfonlayer}{nodelayer}
		\node [style=none] (0) at (-2, 0.25) {};
		\node [style=none] (1) at (-1.5, -0.25) {};
		\node [style=none] (2) at (-1, 0.25) {};
	\end{pgfonlayer}
	\begin{pgfonlayer}{edgelayer}
		\draw [bend right=45, looseness=1.00] (0.center) to (1.center);
		\draw [bend right=45, looseness=1.00] (1.center) to (2.center);
	\end{pgfonlayer}
\end{tikzpicture}} := \ket{00}+\ket{11} \qquad\qquad\qquad
 \intf{\input{wire-cap.tikz}} := \bra{00}+\bra{11}.
\]

The juxtaposition of two diagrams is interpreted as the tensor product of the corresponding matrices and the sequential composition of two diagrams is interpreted as the matrix product of the corresponding matrices:
\[
 \intf{\gendiagram{$D_1$}\;\gendiagram{$D_2$}} := \intf{\gendiagram{$D_1$}}\otimes\intf{\gendiagram{$D_2$}} \qquad\qquad \intf{\input{sequential-composition.tikz}} := \intf{\gendiagram{$D_2$}}\circ\intf{\gendiagram{$D_1$}}
\]

The statements of the relations of the ZH-calculus will be simplified by introducing two derived generators, called \emph{grey spiders} and NOT, respectively.
\beq\label{eq:grey-spider}
 \input{X-spider-dfn.tikz}
\eeq
\beq\label{eq:X-dfn}
 \input{negate-dfn.tikz}
\eeq
With these definitions, \graymult\ acts on computational basis states as XOR and \greyphase{\neg} acts as NOT:
\[
 \intf{\graymult} = \ketbra{0}{00}+\ketbra{0}{11}+\ketbra{1}{01}+\ketbra{1}{10} \qquad\qquad\qquad \intf{\greyphase{\neg}}=\ketbra{0}{1}+\ketbra{1}{0}.
\]



There is an evident encoding of the generators of the ZX-calculus into ZH given by the following translation:
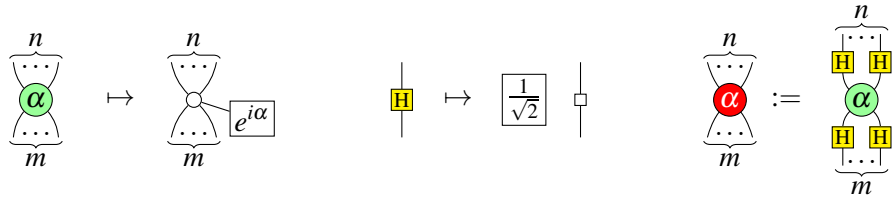
\[
 \input{green-spider.tikz} \qquad\qquad
 \input{Hadamard.tikz} \qquad\qquad
 \input{red-spider.tikz}
\]
Since it is well-known that the ZX-calculus is universal for representing arbitrary linear maps $\mathbb C^{2^m} \to \mathbb C^{2^n}$, the following is immediate:

\begin{proposition}
Any linear map $\mathbb C^{2^m} \to \mathbb C^{2^n}$ can be expressed using the generators of the ZH-calculus.
\end{proposition}

We will also give a normal form in Section~\ref{s:completeness} which directly implies universality of the ZH-calculus, without going via ZX.

\begin{figure}
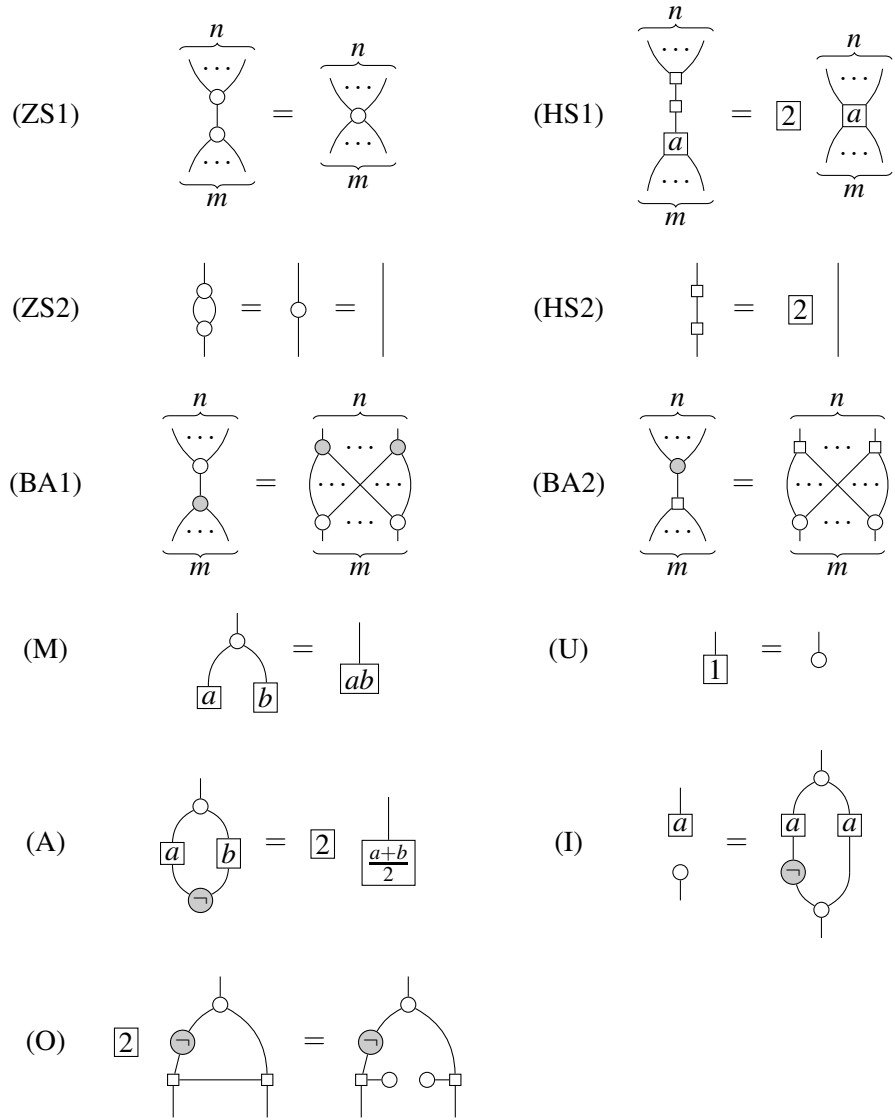

 \centering
 \begin{tabular}{ccccc}
  (ZS1) & \input{Z-spider-rule.tikz} & \qquad & (HS1) & \input{H-spider-rule.tikz} \\ &&&& \\
  (ZS2) & \input{Z-special.tikz} & & (HS2) & \input{H-identity.tikz} \\ &&&& \\
  (BA1) & \input{ZX-bialgebra.tikz} & & (BA2) & \input{ZH-bialgebra.tikz} \\ &&&& \\
  (M) & \input{multiply-rule.tikz} & & (U) & \input{unit-rule.tikz} \\ &&&& \\
  (A) & \input{average-rule.tikz} & & (I) & \input{intro-rule.tikz} \\ &&&& \\
  (O) & \input{ortho-rule.tikz} & & &
 \end{tabular}
 \caption{The rules of the ZH-calculus.
 Throughout, $m,n$ are nonnegative integers and $a,b$ are arbitrary complex numbers.
 The right-hand sides of both \textit{bialgebra} rules (BA1) and (BA2) are complete bipartite graphs on $(m+n)$ vertices, with an additional input or output for each vertex.
 The horizontal edges in equation (O) are well-defined because only the topology matters and we do not need to distinguish between inputs and outputs of generators. n.b. the rules (M), (A), (U), (I), and (O) are pronounced \textit{multiply}, \textit{average}, \textit{unit}, \textit{intro}, and \textit{ortho}, respectively.}
 \label{fig:ZH-rules}
\end{figure}

\subsection{The relations, and soundness}\label{sec:relations}

The rules of the ZH-calculus are given in Figure~\ref{fig:ZH-rules}. We furthermore add one meta rule, often stated as ``only topology matters''. That is, two diagrams are considered equivalent if one can be deformed into the other.
Furthermore, the Z-spiders and H-boxes are assumed to be \emph{symmetric} and \emph{undirected}: i.e.\ two inputs or outputs of the same generator can be swapped without changing the interpretation, and an input can be ``bent'' around to become an output, or conversely. Graphically:
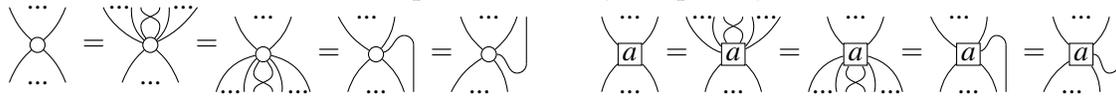
\ctikzfig{generator-symmetries}


\medskip

\begin{proposition}
 The ZH-calculus is sound.
\end{proposition}
\begin{proof}
It is straightforward to check the symmetry properties for each generator and all of the rules in Figure~\ref{fig:ZH-rules} by concrete calculation. Soundness of the meta rule ``only the topology matters'' follows by considering the string diagrams as morphisms in a compact closed category~\cite{SelingerCPM}.
\end{proof}

\subsection{!-box notation}\label{sec:bang-boxes}

Many of the calculations in this paper are greatly simplified by the use of \textit{!-box notation}~\cite{kissinger2014pattern}. A !-box (pronounced ``bang box'') in a string diagram represents a part of the diagram that is able to fan out arbitrarily. That is, the contents of a !-box, along with any wires into or out of the !-box, can be copied $n$ times for any non-negative integer $n$.
For example, the !-box diagram below represents the following family of (concrete) string diagrams, one for each $n$:
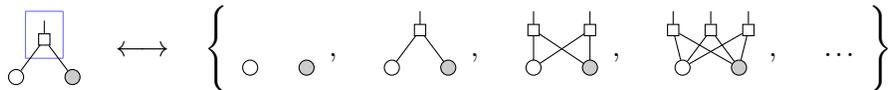
\[ \input{bang-box-example.tikz} \quad \longleftrightarrow \quad
 \left\{
 \ \ \begin{tikzpicture}
	\begin{pgfonlayer}{nodelayer}
		\node [style=white dot] (0) at (-0.75, -0.5) {};
		\node [style=gray dot] (7) at (0.75, -0.5) {};
	\end{pgfonlayer}
\end{tikzpicture}
\ \ ,\quad
 \ \ \input{bang-box-example1.tikz}\ \ ,\quad
 \ \ \input{bang-box-example2.tikz}\ \ ,\quad
 \ \ \input{bang-box-example3.tikz}\ \ ,\quad
 \ \ \ldots\ \  \right\}
\]
All of the resulting string diagrams are well-defined because all of our generators can have arbitrary arities. We can also use !-boxes in diagram equations, as long as each !-box on the LHS has a corresponding !-box on the RHS, and the inputs/outputs in each !-box match. Such a rule represents a family of equations where each \textit{pair} of corresponding !-boxes is replicated $n$ times:
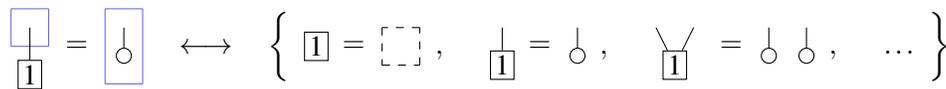
\[
\input{unit-bangboxed.tikz} \quad \longleftrightarrow \quad
\left\{
\ \ \input{unit-bb0.tikz}\ \ ,\quad
\ \ \input{unit-bb1.tikz}\ \ ,\quad
\ \ \input{unit-bb2.tikz}\ \ ,\quad
\ \ \ldots\ \ 
\right\}
\]
Note the dashed box on the right-hand side of the first equation denotes an empty diagram.


With this notation, the definition of grey spiders \eqref{eq:grey-spider} becomes
\beq\label{eq:grey-spider-dfn}
 \input{X-spider-dfn-bb.tikz}
\eeq
Additionally, the rules (ZS1), (HS1), (BA1), and (BA2) from Figure~\ref{fig:ZH-rules} become
\[
 \text{(ZS1)}\quad \input{Z-spider-rule-bb.tikz} \qquad
 \text{(HS1)}\quad \input{H-spider-rule-bb.tikz} \qquad
 \text{(BA1)}\quad \input{ZX-bialgebra-bb.tikz} \qquad
 \text{(BA2)}\quad \input{ZH-bialgebra-bb.tikz}
\]
Using the rules in this form makes it straightforward to prove !-box generalisations of the rules (M), (U), (A), and (I).

\begin{lemma}\label{lem:bb-rules}
The ZH-calculus satisfies the following rules:
 \[
  \text{(M!)}\;\; \input{multiply-rule-bb.tikz} \qquad
  \text{(U!)}\;\; \input{unit-bangboxed.tikz} \qquad
  \text{(A!)}\;\; \input{avg-lemma.tikz} \qquad
  \text{(I!)}\;\; \input{intro-rule-bangboxed.tikz}
 \]
\end{lemma}

\noindent This lemma is proved in Appendix~\ref{sec:bang-rules}. At this point, it is worth highlighting the special cases of (M!) and (U!) where the !-box is expanded $0$ times:
\[
  \begin{tikzpicture}
	\begin{pgfonlayer}{nodelayer}
		\node [style=hadamard] (0) at (-1, -0) {$b$};
		\node [style=hadamard] (1) at (-2, -0) {$a$};
		\node [style=hadamard] (2) at (1, -0) {$ab$};
		\node [style=none] (3) at (0, -0) {$=$};
	\end{pgfonlayer}
\end{tikzpicture} \qquad\qquad\qquad\qquad\qquad\qquad
  \input{scalar-rule.tikz}
\]
These rules enable us to multiply scalars at will, and in particular eliminate scalars by multiplying by the inverse. From hence forth, we will use this fact without further comment during our proofs.

In this paper, we use a mild, but very useful, extension of the usual !-box notation, which allows !-boxes to be indexed by a the elements of a finite set. For example, indexing over the finite set $\mathbb B^2 := \{ 00, 01, 10, 11 \}$, we can write expressions such as:
\[
\input{indexed-example.tikz}
\ \ :=\ \ \ 
\input{index-example-rhs.tikz}
\]
This extends to equations in the obvious way:
\[
\left(\ \input{index-example-rule.tikz}\ \right)
\ \ := \ \ 
\left( \ \input{index-example-rule-inst.tikz}\ \right)
\]
where we require corresponding !-boxes on the LHS and RHS to be indexed by the \textit{same} finite set. Note that inputs and outputs of a copy associated with the index $x \in X$ on the LHS are matched with inputs and outputs of the \textit{same} copy on the RHS.

We recover the behaviour of normal, un-labelled !-boxes by interpreting a !-box without a label as being indexed by an \textit{arbitrary} finite set, e.g.
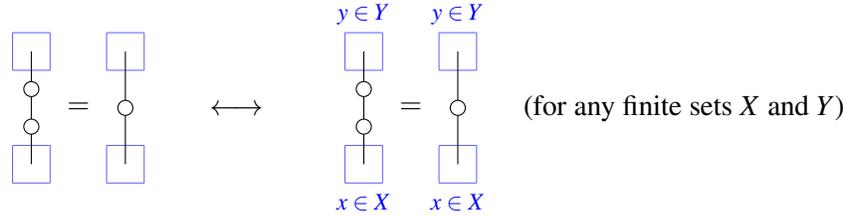
\[
\input{Z-spider-rule-bb.tikz} \qquad
\longleftrightarrow \qquad
\input{Z-spider-rule-bb-index.tikz} \quad \textrm{(for any finite sets $X$ and $Y$)} \]

\section{Completeness}
\label{s:completeness}

We show that the ZH-calculus is complete by demonstrating the existence of a unique normal form for ZH-diagrams. It is first worth noting that, because we can turn inputs into outputs arbitrarily (cf. the beginning of section~\ref{sec:relations}), it suffices to consider diagrams which have only outputs. We call these \textit{states}. Concretely, these are interpreted as column vectors (i.e. kets).



For states $\psi,\phi$, let $\psi * \phi$ be the \textit{Schur product} of $\psi$ and $\phi$ obtained by plugging the $i$-th output of $\psi$ and $\phi$ into \whitemult, for each $i$:
\ctikzfig{schur}
It follows from (ZS1) that $*$ is associative and commutative, so we can write $k$-fold Schur products $\psi_1 * \psi_2 * \ldots * \psi_k$ without ambiguity.
For any finite set $J$ with $|J| = k$, let $\prod_{j\in J} \psi_j$ be the $k$-fold Schur product.

Let $\mathbb B^n$ be the set of all $n$-bitstrings. For any $\vec{b} := b_1\ldots b_n \in \mathbb B^n$, define the \textit{indexing map} $\iota_{\vec{b}}$ as follows:
\begin{equation}\label{eq:iota-dfn}
\iota_{\vec{b}} \; = \;
\input{indexing-box.tikz} \; = \; \left(\greyphase{\neg}\right)^{1 - b_1} \ldots \left(\greyphase{\neg}\right)^{1 - b_n}.
\end{equation}
Then normal forms are given by the following $2^n$-fold Schur products:
\begin{equation}\label{eq:nf-formula}
  \prod_{\vec{b} \in \mathbb B^n} \big( \iota_{\vec{b}} \circ H_n(a_{\vec{b}}) \big)
\end{equation}
where $H_n(a_{\vec{b}})$ is the arity-$n$ H-box (considered as a state) labelled by an arbitrary complex number $a_{\vec{b}}$.

A normal form diagram can be seen as a collection of $n$ spiders, fanning out to $2^n$ H-boxes, each with a distinct configuration of NOT's corresponding to the $2^n$ bitstrings in $\mathbb B^n$. Diagrammatically, normal forms are:
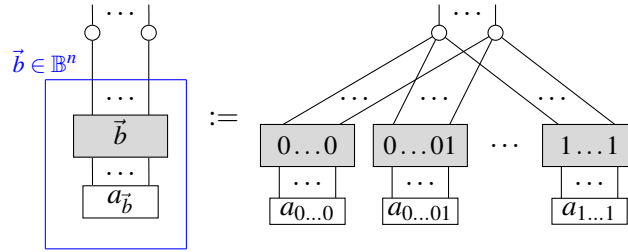
\[
\input{nf-bbox.tikz}\ \ :=\ \ 
\input{nf-picture.tikz}
\]

\begin{theorem}\label{thm:nf-unique}
Normal forms are unique. In particular:
\begin{equation}\label{eq:nf-concrete}
\intf{ \, \prod_{\vec{b} \in \mathbb B^n} \big( \iota_{\vec{b}} \circ H_n(a_{\vec{b}}) \big) } =
\sum_{\vec{b} \in \mathbb B^n} a_{\vec{b}} \ket{\vec{b}}.
\end{equation}
\end{theorem}

\begin{proof}
  The map $\iota_{\vec b}$ is a permutation that acts on computational basis elements as $\ket{\vec c} \mapsto \ket{\vec c \oplus \vec b \oplus \vec 1}$. In particular, it sends the basis element $\ket{\vec 1}$ to $\vec b$. Hence $\iota_{\vec b} \circ H_n(a_{\vec b})$ is a vector with $a_{\vec b}$ in the $\vec b$-th component and $1$ everywhere else. The Schur product of all such vectors indeed gives the RHS of ~\eqref{eq:nf-concrete}.
\end{proof}

Since equation~\eqref{eq:nf-concrete} gives us a means of constructing any vector in $\mathbb C^{2^n}$, Theorem~\ref{thm:nf-unique} can also be seen as a proof of universality of the ZH calculus, independent of the encoding into ZX we gave in section~\ref{s:ZX-translation}.




 


We now prove 2 lemmas which will assist in manipulating normal forms:

\begin{lemma}\label{lem:X-copy}
 The NOT operator copies through white spiders:
 \ctikzfig{X-copy}
\end{lemma}
\begin{proof}
 Starting from the left-hand side,
 \[
  \input{X-copy-proof.tikz} \qedhere
 \]
\end{proof}

\begin{lemma}\label{lem:iota-copy}
 The $\iota_{\vec{b}}$ operator copies through white spiders, i.e.\ for any $\vec{b}\in\mathbb B^n$:
 \ctikzfig{iota-copy}
\end{lemma}
\begin{proof}
 This follows immediately from Lemma~\ref{lem:X-copy} via the definition of $\iota_{\vec{b}}$ \eqref{eq:iota-dfn}.
\end{proof}

\begin{lemma}\label{lem:convolution-iota}
 The ZH-calculus enables the computation of the Schur product of two maps of the form $\iota_{\vec{b}}\circ H_n(x)$ and $\iota_{\vec{b}}\circ H_n(y)$ for any $\vec{b}\in\mathbb B^n$ and $x,y\in\mathbb C$:
 \ctikzfig{convolution-iota}
\end{lemma}
\begin{proof}
Apply Lemma~\ref{lem:iota-copy}, followed by (M!).
\end{proof}


We will now show that normal form diagrams, when combined in various ways, can also be put into normal form. Let \input{nf.tikz} denote an arbitrary normal-form diagram. It is straightforward to see that permuting the outputs of a normal-form diagram merely interchanges the bits in the coefficients $a_{\vec b}$. Hence, normal forms are preserved under permutations of outputs. Furthermore:

\begin{proposition}\label{prop:extension}
 A diagram consisting of a normal form diagram juxtaposed with \whiteunit\ can be brought into normal form using the rules of the ZH-calculus:
 \ctikzfig{extension}
\end{proposition}
\begin{proof}
 Starting from the left-hand side, which we expand using the indexed !-box notation,
 \ctikzfig{extension-proof}
 The last diagram is a normal form diagram with $n+1$ outputs, i.e.\ the desired result.
\end{proof}

\begin{proposition}\label{prop:convolution}
 The Schur product of two normal form diagrams can be brought into normal form using the rules of the ZH-calculus.
 \ctikzfig{convolution-nf}
\end{proposition}
\begin{proof}
 This follows from (ZS1) and Lemma~\ref{lem:convolution-iota}.
\end{proof}

\begin{corollary}\label{cor:tensor-product}
 The tensor product of two normal form diagrams can be brought into normal form using the rules of the ZH-calculus.
\end{corollary}
\begin{proof}
 A tensor product can be expressed as
 \ctikzfig{tensor-product}
 The diagram NF$_1$ and the leftmost $m$ copies of \whiteunit\ can be combined into one normal-form diagram with $(n+m)$ outputs by successive applications of Proposition~\ref{prop:extension}.
 Similarly, the rightmost $n$ copies of \whiteunit\ and NF$_2$ can be combined into one normal-form diagram with $(n+m)$ outputs.
 The desired result then follows by Proposition~\ref{prop:convolution}.
\end{proof}

\begin{remark}\label{rem:scalar-juxtaposition}
 Note that a single scalar H-box is a normal form diagram.
 Corollary~\ref{cor:tensor-product} thus implies that a diagram consisting of a normal form diagram juxtaposed with a scalar H-box can be brought into normal form.
 In the following proofs, we will therefore ignore scalars for simplicity: they can be added back in and then incorporated to the normal form without problems.
\end{remark}

We are now ready to prove the most difficult case, which is contraction. The majority of the work goes into proving Lemma~\ref{lem:big-disconnect}, which we call the Disconnect Lemma. It uses the (O) rule to disconnect the $2^n$-legged $\whitedot$-spider arising from a contraction of a normal form into $2^{n-1}$ separate cups. It was proven with the help of the graphical proof assistant Quantomatic. Details and full proof are given in Appendix~\ref{sec:disconnect}.

\begin{proposition}\label{prop:contraction}
 The diagram resulting from applying \whitecounit\ to an output of a normal form diagram can be brought into normal form:
 \ctikzfig{whitecounit-nf}
\end{proposition}

\begin{proof}
  Starting from an arbitrary normal form, with a \whitecounit plugged into the right most output, we have:
  \[ \scalebox{0.8}{\input{contraction-thm-pf.tikz}} \]
  Then, we can apply Lemma~\ref{lem:big-disconnect}:
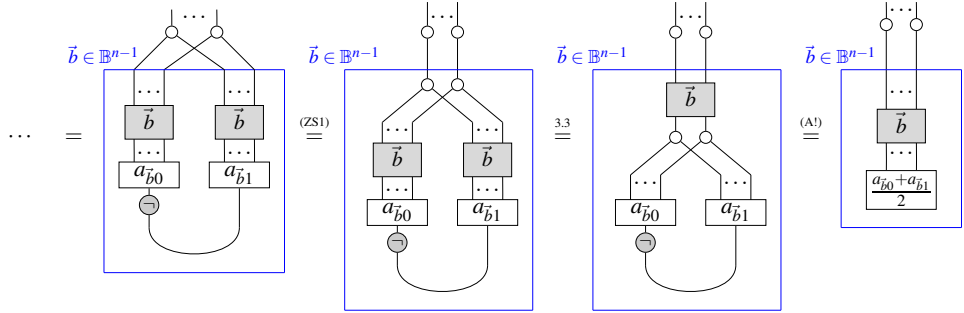
  \[ \scalebox{0.8}{\input{contraction-thm-pf2.tikz}} \]
  The final diagram is in normal form, which completes the proof.
\end{proof}

Our strategy will now be to show that any diagram can be decomposed into H-boxes, combined via the operations of extension, convolution, and contraction. This will give us a completeness proof, thanks to the following proposition.

\begin{lemma}\label{lem:H-box-nf}
 Any H-box can be brought into normal form using the rules of the ZH-calculus.
\end{lemma}
\begin{proof}
The matrix of an H-box $H_n(a)$ has 1's in every entry but the very last one. Hence, to bring an H-box into normal form, we just need to introduce `dummy' 1's for every other matrix entry. We demonstrate the principle using a binary H-box but the argument is analogous for any other arity:
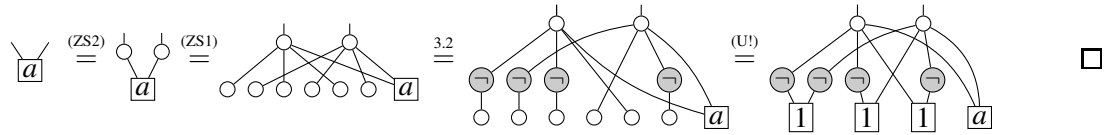
\[
\input{H-nf-example.tikz} \qedhere
\]
\end{proof}

To simplify the decomposition of diagrams into H-boxes, we prove a few corollaries.

\begin{corollary}\label{cor:cup-nf}
The diagram of a single cup can be brought into normal form:
\[ \input{cup-nf.tikz} \]
\end{corollary}

\begin{proof}
We can rewrite the cup as a pair of H-boxes using (HS2). This can then be written in terms of extension, convolution, and contraction as follows:
\ctikzfig{binary-Z-decomposition}
Hence, we can apply Lemma~\ref{lem:H-box-nf} and Propositions \ref{prop:extension}, \ref{prop:convolution}, and \ref{prop:contraction} to get a normal form.
\end{proof}

\begin{corollary}\label{cor:whitemult-nf}
 The diagram resulting from applying \whitemult\ to a pair of outputs of a normal form diagram can be brought into normal form.
 \begin{equation}\label{eq:whitemult-nf}
 	\input{whitemult-nf.tikz}
 \end{equation}
\end{corollary}

\begin{proof}
Applying a \whitemult\ to a pair of outputs has the same result as convolving with a cup, then contracting one of the outputs. That is, we can decompose \eqref{eq:whitemult-nf} as follows:
\ctikzfig{whitemult-decomp}
then apply Corollary \ref{cor:cup-nf} and Propositions \ref{prop:extension}, \ref{prop:convolution}, and \ref{prop:contraction}.
\end{proof}

\begin{corollary}\label{cor:cap-nf}
Applying a cap to a normal form diagram results in another normal form diagram:
\ctikzfig{cap-nf}
\end{corollary}

\begin{proof}
Since the cap can be decomposed as $\whitecounit \circ \whitemult$, the result follows immediately from Corollary~\ref{cor:whitemult-nf} and Proposition~\ref{prop:contraction}.
\end{proof}

Thanks to Corollaries~\ref{cor:tensor-product} and \ref{cor:cap-nf}, we are able to turn any diagram of normal forms into a normal form. It only remains to show that the generators of the ZH-calculus can themselves be made into normal forms. We have already shown the result for H-boxes, so we only need the following.

\begin{lemma}\label{lem:Z-spider-nf}
Any Z-spider can be brought into normal form using the rules of the ZH-calculus.
\end{lemma}
\begin{proof}
We can turn \whiteunit{} into an H-box using (U) and then bring it into normal form via Lemma~\ref{lem:H-box-nf}.
By (ZS1), $\whitedot = \input{dot-nf.tikz}$, which can be brought into normal form using (U), Lemma~\ref{lem:H-box-nf}, and Corollaries~\ref{cor:tensor-product} and \ref{cor:cap-nf}.
This covers the cases of Z-spiders with 0 or 1 incident wires.

We can decompose any Z-spider with $n\geq 2$ incident wires as a tensor product of $(n-1)$ cups, with each cup \whitemult-ed with its neighbours:
\ctikzfig{n-ary-Z-decomposition}
If $n=2$, no \whitemult are needed and the equality is by (ZS2) instead of (ZS1).
In either case, the diagram can be brought into normal form by applying Corollaries~\ref{cor:tensor-product}, \ref{cor:cup-nf}, and \ref{cor:whitemult-nf}.
\end{proof}

\begin{theorem}
The ZH-calculus is complete: for any ZH diagrams $D_1$ and $D_2$, if $\llbracket D_1 \rrbracket = \llbracket D_2 \rrbracket$ then $D_1$ is convertible into $D_2$ using the rules of the ZH-calculus.
\end{theorem}

\begin{proof}
By Theorem~\ref{thm:nf-unique}, it suffices to show that any ZH diagram can be brought into normal form. Lemmas~\ref{lem:H-box-nf} and \ref{lem:Z-spider-nf} suffice to turn any generator into normal form. Corollary~\ref{cor:tensor-product} lets us turn any tensor product of generators into a normal form and Corollary~\ref{cor:cap-nf} lets us normalise any arbitrary wiring.
\end{proof}

\section{Applications and future work}\label{s:applications}

We will now briefly survey some of the potential applications for the ZH-calculus. We begin with the simple observation that $n$-ary H-boxes let us generalise the usual string diagrammatic description the controlled-Z gate (as in e.g. the ZX calculus) to an $n$-controlled-Z gate:
\ctikzfig{n-controlled-Z}
Using the decomposition of controlled-Z gates above, a representation of graph states as ZX-diagrams was given in~\cite{DP1}, which in turn gave a fully diagrammatic derivation of the local complementation law for graph states~\cite{DP1} and a new procedure for extracting circuits from certain computations in the one-way model~\cite{DP2}. Passing from $\wedge Z$ to $\wedge^n Z$ gives an analogous representation for \textit{hypergraph states}:
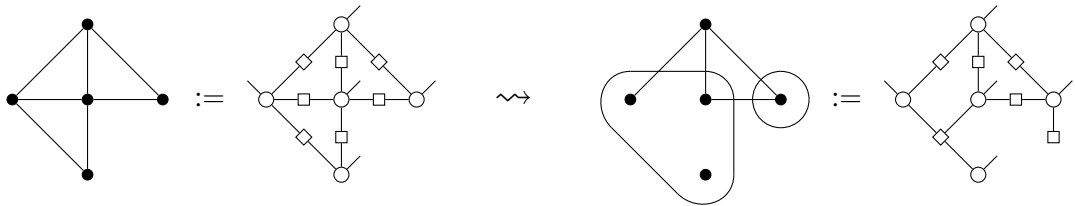
\[ \input{gs-graph-s.tikz}
   \qquad\textrm{\Large $\leadsto$}\qquad
   \input{gs-hypergraph.tikz}
\]
Indeed this was the original motivation for considering H-boxes of arbitrary arity. Using a method analogous to proofs in Appendix~\ref{sec:bang-rules}, we can routinely introduce !-boxes to known rules involving graph states (e.g. local complementation and feed-forward rules) to generalise them to hypergraph states. For example, introducing !-boxes to the local complementation rule enables complementing hyperedges of arbitrary arity overlapping on a single vertex: 
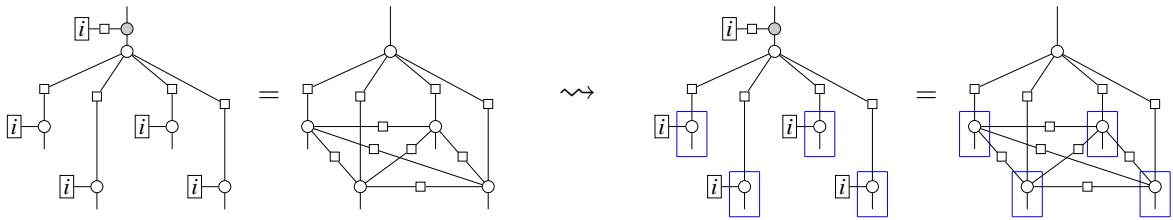
\[
\scalebox{0.8}{\input{lc1.tikz}} \ \  = \ \  \scalebox{0.8}{\input{lc2.tikz}}
\qquad\textrm{\Large $\leadsto$}\qquad
\scalebox{0.8}{\input{lc-bb1.tikz}} \ \  = \ \  \scalebox{0.8}{\input{lc-bb2.tikz}}
\]
This potentially gives a powerful new language and set of techniques for working with hypergraph states.
Exploring these techniques, and the relationship to known rules for manipulating hypergraph states is a topic of future work.

In another direction, if we consider diagrams whose H-boxes are labelled by a fixed root of unity $\omega := \exp(i \pi/2^m)$, we obtain an encoding for unitary gates described by arbitrary \textit{phase polynomials}~\cite{moscamatroid}, i.e. gates of the form $U_{\phi} \ket{\vec b} = \omega^{\phi(\vec b)} \ket{\vec b}$ for some polynomial $\phi(\vec b)$ over $n$ boolean variables. These have a simple graphical representation, where Z-spiders represent variables and $\omega$-labelled H-boxes represent terms in the phase polynomial. For example:
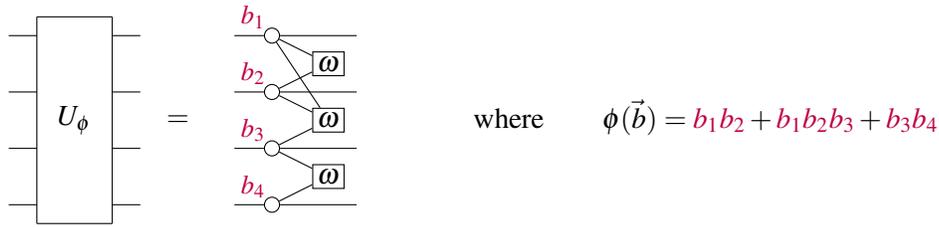
\[
\input{phase-poly.tikz}
\qquad\qquad
\textrm{where}\qquad
\phi(\vec b) = {\color{purple} b_1 b_2} + {\color{purple} b_1 b_2 b_3} + {\color{purple} b_3 b_4}
\]
One can then straightforwardly show basic properties of these unitaries (e.g.~composition, commutation, and replacement of non-linear AND terms by linear XOR terms) using the rules of the ZH-calculus.
The phase polynomial formalism for $m = 2$ has been used extensively in studying optimisation problems for Clifford+T circuits~\cite{MeetInMiddle,moscamatroid,AmyMoscaReedMuller,campbelltcount}, and it was recently shown that all diagonal gates in the Clifford hierarchy are of the form $U_\phi$, where the level of the hierarchy depends on $m$ and the degree of $\phi$~\cite{DiagHierarchy}. Gaining access to this phase polynomial structure diagrammatically could therefore yield new methods for quantum circuit optimisation and/or fault tolerant computation through automated diagram rewriting in tools like Quantomatic.

\bigskip

\noindent \textbf{Acknowledgements.}
The authors would like to thank Simon Perdrix and Mariami Gachechiladze for the fruitful conversations in which the foundations of the ZH-calculus were developed.
We are also grateful to Niel de Beaudrap for interesting discussions about applications of the ZH-calculus and to Sal Wolffs for careful reading of our proofs (and pointing out a major omission in Corollary~\ref{cor:whitemult-nf}).

The research leading to these results has received funding from the European Research Council under the European Union's Seventh Framework Programme (FP7/2007-2013) ERC grant agreement no.\ 334828 (Backens) and 320571 (Kissinger).
The paper reflects only the authors' views and not the views of the ERC or the European Commission.
The European Union is not liable for any use that may be made of the information contained therein.


\bigskip

\bibliographystyle{eptcs}
\bibliography{main}

\newpage

\appendix

\section{Proofs of the !-box versions of basic rules}\label{sec:bang-rules}

The following Lemma will be useful for the later results.

\begin{lemma}\label{lem:half-times-2}
 The scalar H-box with phase 2 and the scalar H-box with phase $\frac{1}{2}$ are inverses of each other.
 \ctikzfig{half-times-2}
\end{lemma}
\begin{proof}
 Starting from the left-hand side, we have
 \[
   \input{half-times-2-proof.tikz} \qedhere
 \]
\end{proof}

\begin{lemma}\label{prop:multiply-bb}
 The following !-box generalisation of (M) holds in the ZH-calculus.
 \ctikzfig{multiply-rule-bb}
\end{lemma}
\begin{proof}
 Starting from the left-hand side, we have
 \ctikzfig{multiply-rule-bb-proof1}
 \[
  \input{multiply-rule-bb-proof2.tikz} \qedhere
 \]
\end{proof}


\begin{lemma}\label{lem:unit-bb}
 The phase-1 H-box of any arity decomposes:
 \ctikzfig{unit-bangboxed}
\end{lemma}
\begin{proof}
 Again, starting from the left-hand side,
 \[
  \input{unit-bb-proof.tikz} \qedhere
 \]
\end{proof}

As we noted at the end of Section~\ref{sec:bang-boxes}, Lemmas~\ref{prop:multiply-bb} and \ref{lem:unit-bb} allow us to combine and cancel scalars at will. We will do this for the remainder of the proofs.

\begin{lemma}\label{lem:avg-bang}
 The following !-box generalisation of (A) holds in the ZH-calculus.
 \ctikzfig{avg-lemma}
\end{lemma}
\begin{proof}
  Starting from the left-hand side,
  \[ \input{avg-lemma-pf.tikz} \qedhere \]
\end{proof}

\begin{lemma}\label{lem:intro-bb}
 The introduction rule (I) holds for H boxes with arbitrarily many inputs:
 \ctikzfig{intro-rule-bangboxed}
\end{lemma}
\begin{proof}
 As usual, we start from the left-hand side:
 \[
  \input{intro-rule-bb-proof.tikz} \qedhere
 \]
\end{proof}

Combining Lemmas~\ref{prop:multiply-bb}--\ref{lem:intro-bb} gives a proof of Lemma~\ref{lem:bb-rules}.

\section{Proof of the disconnect lemma}\label{sec:disconnect}

The proof of the disconnect lemma was produced with the aid of Quantomatic. In particular, Lemmas~\ref{lem:disconnect-4} and \ref{lem:big-disconnect} are exported from the Quantomatic derivations \texttt{disconnect-4} and \texttt{gen-disconnect-4} in the \texttt{zh} project, available here:
\begin{center}
	\color{blue}\url{https://github.com/Quantomatic/sample-projects/tree/qpl2018}
\end{center}
It should be viewable in the latest version of Quantomatic. However, for the sake of posterity, the version of \texttt{Quantomatic.jar} used to make this project has also been uploaded here:
\begin{center}
	\color{blue}\url{https://github.com/Quantomatic/sample-projects/releases/tag/qpl2018}
\end{center}

{
\renewcommand{\pi}{\neg}

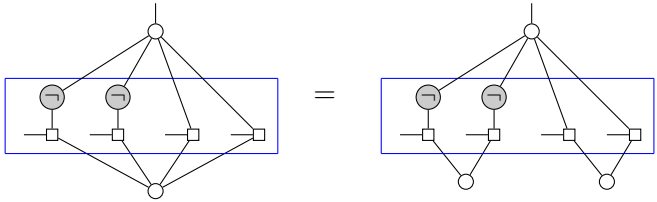
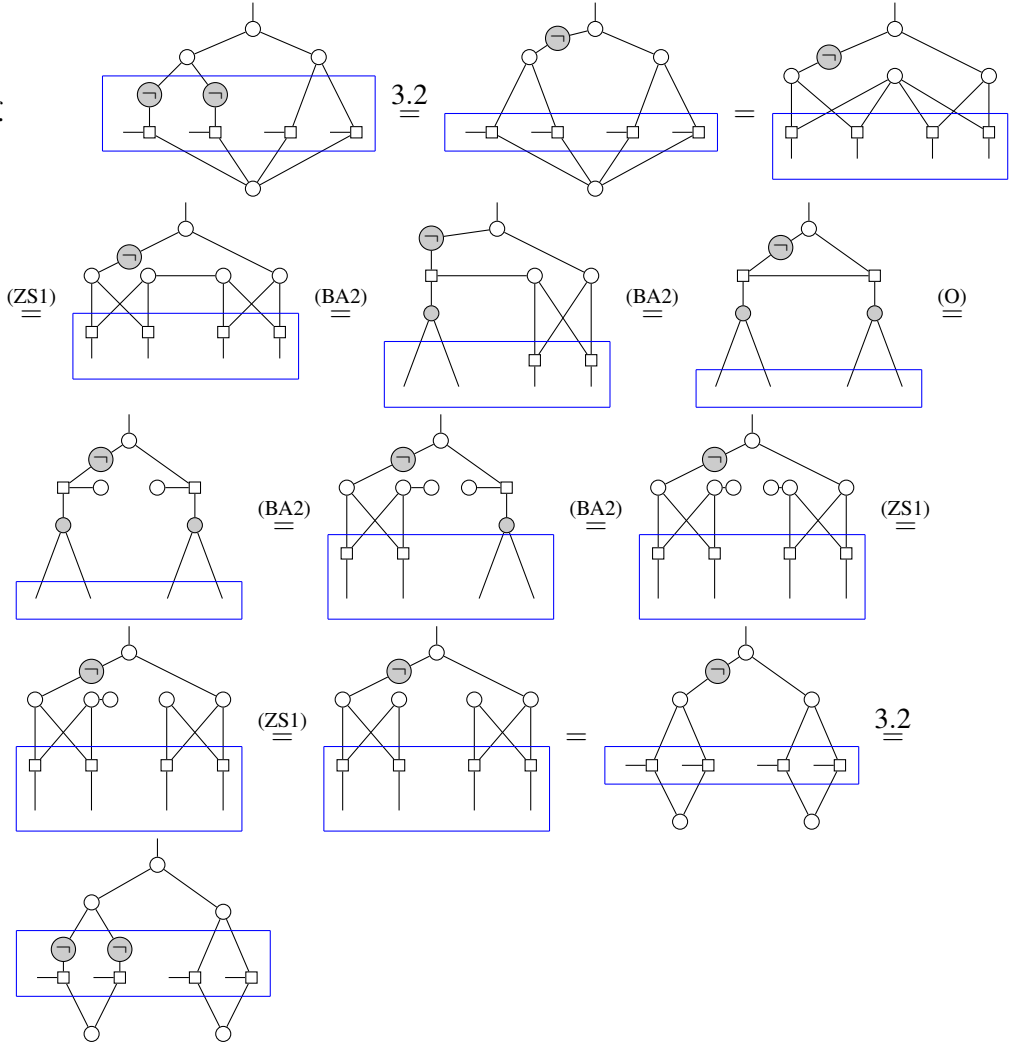
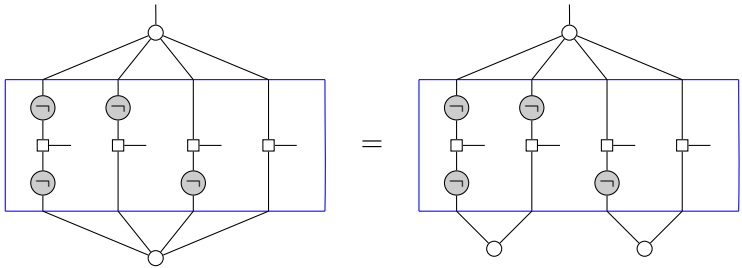
\begin{lemma}\label{lem:disconnect-4}
  \input{disconnect-lemma.tikz}
\end{lemma}

\begin{proof}
  \begin{quote}\raggedright
\input{disconnect-4-pf-0.tikz} $\overset{\ref{lem:X-copy}}{=}$
\input{disconnect-4-pf-1.tikz} $=$
\input{disconnect-4-pf-2.tikz} $\overset{\text{(ZS1)}}{=}$
\input{disconnect-4-pf-3.tikz} $\overset{\text{(BA2)}}{=}$
\input{disconnect-4-pf-4.tikz} $\overset{\text{(BA2)}}{=}$
\input{disconnect-4-pf-5.tikz} $\overset{\text{(O)}}{=}$
\input{disconnect-4-pf-6.tikz} $\overset{\text{(BA2)}}{=}$
\input{disconnect-4-pf-7.tikz} $\overset{\text{(BA2)}}{=}$
\input{disconnect-4-pf-8.tikz} $\overset{\text{(ZS1)}}{=}$
\input{disconnect-4-pf-9.tikz} $\overset{\text{(ZS1)}}{=}$
\input{disconnect-4-pf-10.tikz} $=$
\input{disconnect-4-pf-11.tikz} $\overset{\ref{lem:X-copy}}{=}$
\input{disconnect-4-pf-12.tikz}
\end{quote}

\end{proof}

\begin{lemma}\label{lem:disconnect-step}
  \input{contraction-lemma.tikz}
\end{lemma}

\begin{proof}
 We begin by decomposing the bottom two \greyphase{\neg} nodes on the left-hand side using \eqref{eq:X-dfn}. Then:
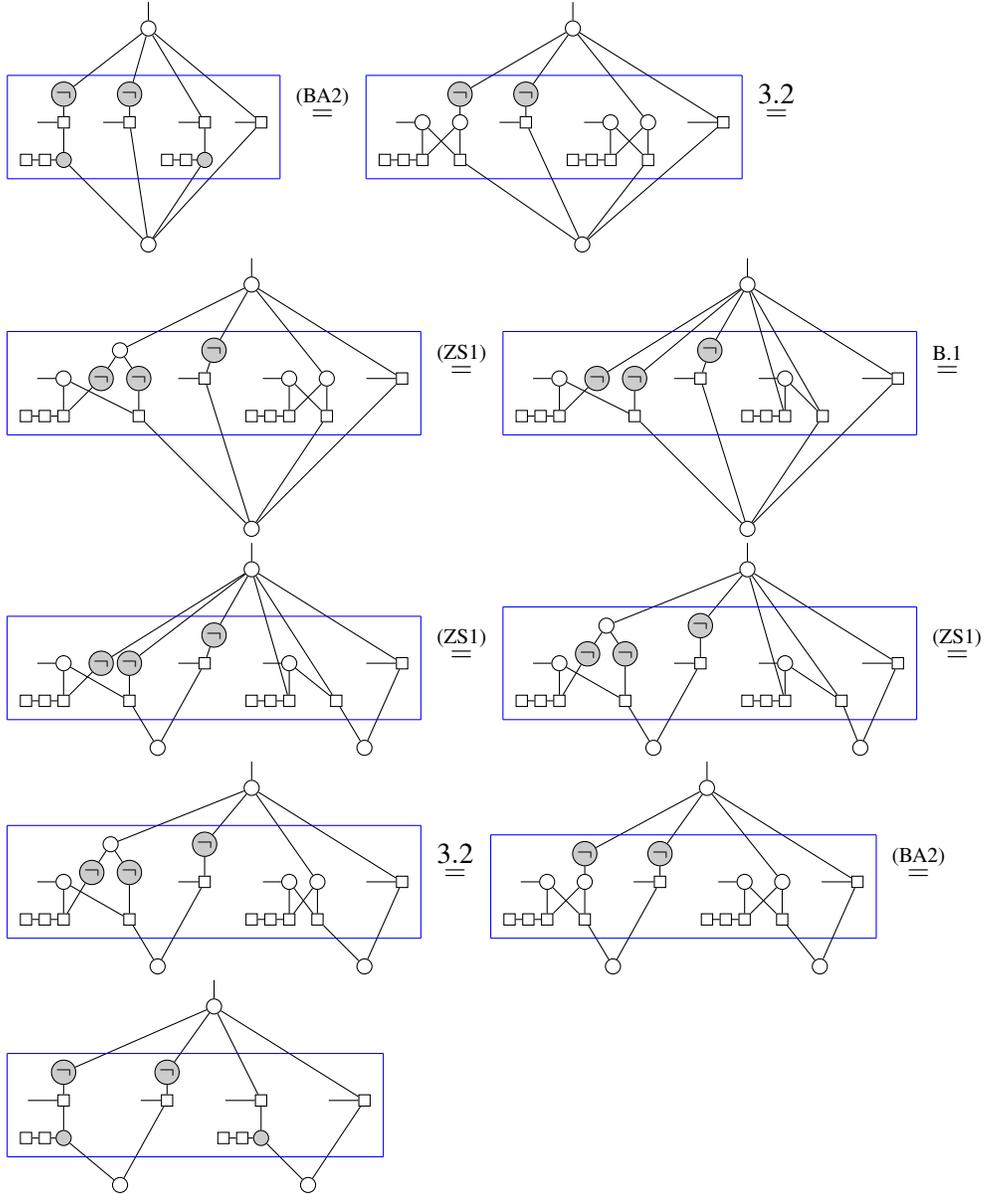
  \begin{quote}\raggedright
\input{pi-disconnect-4-pf-0.tikz} $\overset{\text{(BA2)}}{=}$
\input{pi-disconnect-4-pf-1.tikz} $\overset{\ref{lem:X-copy}}{=}$
\input{pi-disconnect-4-pf-2.tikz} $\overset{\text{(ZS1)}}{=}$
\input{pi-disconnect-4-pf-3.tikz} $\overset{\text{\ref{lem:disconnect-4}}}{=}$
\input{pi-disconnect-4-pf-4.tikz} $\overset{\text{(ZS1)}}{=}$
\input{pi-disconnect-4-pf-5.tikz} $\overset{\text{(ZS1)}}{=}$
\input{pi-disconnect-4-pf-6.tikz} $\overset{\ref{lem:X-copy}}{=}$
\input{pi-disconnect-4-pf-7.tikz} $\overset{\text{(BA2)}}{=}$
\input{pi-disconnect-4-pf-8.tikz}
\end{quote}

 This yields the desired result after reintroducing the \greyphase{\neg} nodes via \eqref{eq:X-dfn}.
\end{proof}

}

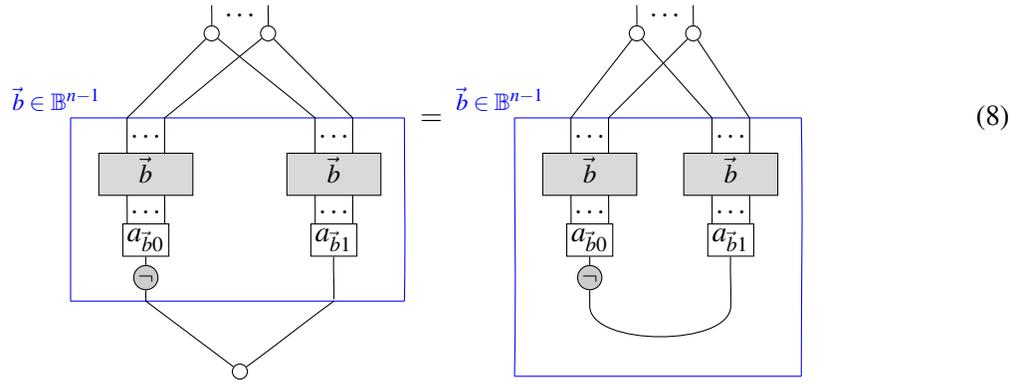
\begin{lemma}\label{lem:big-disconnect}
  \begin{equation}\label{eq:contraction-sep}
  	\input{contraction-sep.tikz} \ =\ \input{contraction-sep-rhs.tikz}
  \end{equation}
\end{lemma}

\begin{proof}
First, note that we can split the set $\mathbb B^{n-1}$ into two pieces, based on whether the most significant bit is $0$ or $1$:
\[
\mathbb B^{n-1} = \{ 0 \vec{b} \ |\ \mathbb B^{n-2}\} \uplus \{ 1 \vec{b} \ |\ \mathbb B^{n-2} \}
\]
Hence, the LHS can be equivalently written as:
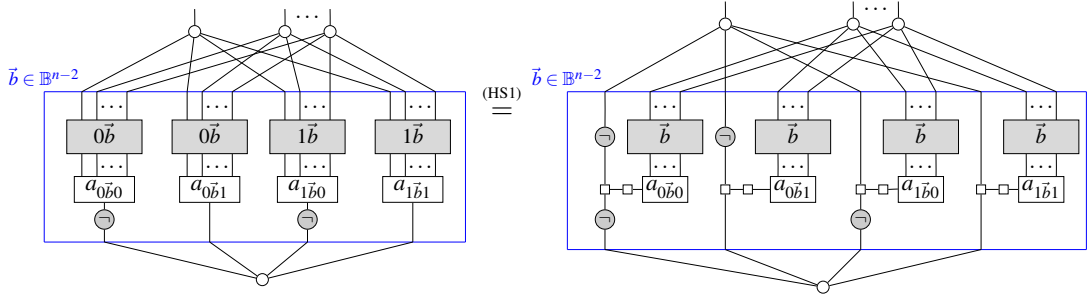
\[
\scalebox{0.8}{\input{contraction-msb-index1.tikz}}
\ \namedeq{(HS1)}\ 
\scalebox{0.8}{\input{contraction-msb-index2.tikz}}
\]
We can then apply Lemma~\ref{lem:disconnect-step} to split the bottom spider in two:
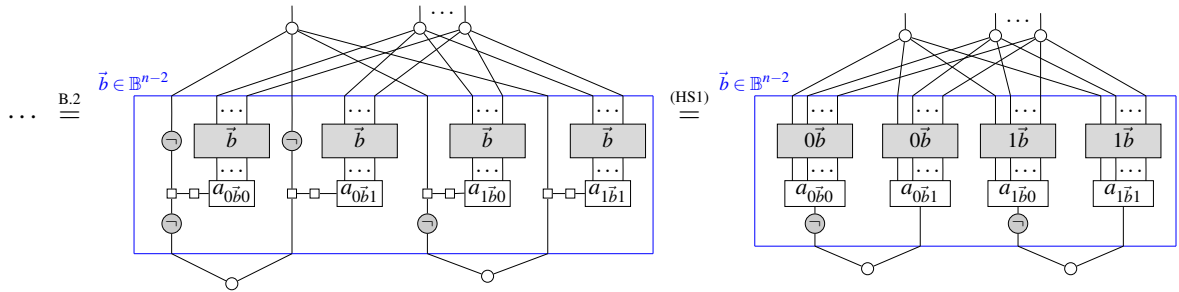
\[
\cdots \ \ \namedeq{\ref{lem:disconnect-step}}\ \ 
\scalebox{0.8}{\input{contraction-msb-index3.tikz}}
\ \namedeq{(HS1)}\ 
\scalebox{0.8}{\input{contraction-msb-index3p.tikz}}
\]
We can then split the !-box into two parts, indexed over the same set to obtain:
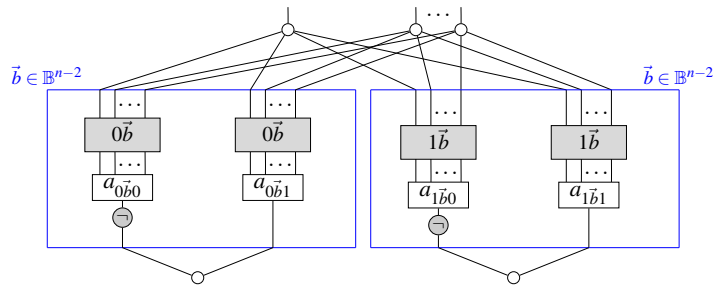
\[
\scalebox{0.8}{\input{contraction-msb-index4.tikz}}
\]
We now have two copies of a graph which is very similar to the LHS of \eqref{eq:contraction-sep}, but for one fewer bit. We can thus repeat the process above to split each of the two spiders using the second bit of the bitstring, then the third, and so on, until $\whitedot$-spiders only connect pairs of $H$-spiders that disagree on the least significant bit. Replacing 2-legged $\whitedot$-spiders with cups, we obtain the RHS of \eqref{eq:contraction-sep}.
\end{proof}

\end{document}